\documentclass[preprint,11pt,authoryear]{elsarticle}
\usepackage{setspace}
\usepackage[margin=1in]{geometry}
\usepackage{amsmath,amssymb,amsthm,bm}
\usepackage{booktabs,graphicx,multirow,rotating}
\usepackage[table]{xcolor}
\usepackage[colorlinks=true, linkcolor=blue, citecolor=blue, urlcolor=blue]{hyperref}
\newtheorem{definition}{Definition}

\newtheorem{lemma}{Lemma}
\newtheorem{theorem}{Theorem}
\newtheorem{proposition}{Proposition}

\theoremstyle{remark}
\newtheorem{remark}{Remark}

\newcommand{\dnorm}[1]{\lVert #1 \rVert_{*}}
\newcommand{\norm}[1]{\lVert #1 \rVert}

\begin{document}

\begin{frontmatter}

\title{Multi-period Mean-Expectile Portfolio Optimization under Wasserstein
Ambiguity: Reformulation, Degeneracy and the Role of the Ground Metric}

\author{Rupendra Yadav\corref{cor1}}
\ead{rupendra.yadav@maths.iitd.ac.in}
\author{Aparna Mehra}
\ead{apmehra@maths.iitd.ac.in}
\cortext[cor1]{Corresponding author.}
\address{Department of Mathematics, Indian Institute of Technology Delhi, Hauz Khas, New Delhi 110016, India}


\begin{abstract}
Expectiles are the only law-invariant risk measures that are both coherent and elicitable. Unlike Conditional Value-at-Risk (CVaR), however, they do not admit a Rockafellar--Uryasev representation that admits tractable Wasserstein reformulations. We address this difficulty by developing an envelope theorem for worst-case expectiles that characterizes the worst-case expectile over a Wasserstein ambiguity set as the unique root of a worst-case expectation with a two-piece affine integrand. This representation permits direct application of standard Wasserstein duality.

Using this result, we reformulate a multi-period tri-level mean--expectile portfolio problem as four parametric linear programs with \(O(N)\) constraints. We establish four structural properties of the proposed model: an endogenously damped price of robustness, a decision-dependent critical radius beyond which the expectile tail component becomes inactive, exact recovery of the nominal model at zero ambiguity, and a characterization of how the Wasserstein ground metric determines whether the limiting portfolio becomes more concentrated or more diversified. Numerical experiments on 90 FTSE constituents over 3,341 out-of-sample trading days show that the expectile model outperforms a CVaR model matched on ambiguity set, radius, ground metric, and trade-off weight in all nine parameter cells---significantly so whenever the radius is non-trivial. The experiments further confirm the predicted degeneracy under the \(\ell_\infty\) ground metric.
\end{abstract}

\begin{keyword}
Portfolio optimization \sep Expectile risk measure \sep
Distributionally robust optimization \sep Wasserstein metric \sep Multi-period
investment
\end{keyword}

\end{frontmatter}

\section{Introduction}
\label{sec:intro}

\citet{Markowitz1952} established the quantitative foundation of modern portfolio theory by formalizing the critical trade-off between risk and return within a static single-period framework. Although foundational, this paradigm often misaligns with realistic market behaviors that require continuous wealth reallocation across sequential horizons to capture and compound long-term cumulative gains \citet{li2022multi, yan2026wasserstein, wu2023distributionally}. Moving from a single-period to a multi-period framework is not merely a matter of repeating the same optimization problem. It also raises questions concerning the intertemporal evaluation of risk and the consistency of sequential allocations \citep{cui2019time,fiessinger2025time}. These issues become particularly relevant when the return distribution itself is uncertain, since both the portfolio allocation and its risk assessment depend on an estimated probability law. Within this dynamic framework, the choice of an appropriate risk measure becomes particularly important, as it determines how portfolio losses are evaluated across successive investment periods and directly influences the resulting portfolio allocations.

Value-at-Risk (VaR) remains widely used in financial practice, but its lack of subadditivity and its inability to account for the magnitude of losses beyond the specified quantile are well known \citep{jorion1996risk2,artzner1999coherent}. Conditional Value-at-Risk (CVaR) addresses these shortcomings by averaging losses beyond the VaR threshold and, more importantly from an optimization perspective, admits the representation of \citet{rockafellar2000optimization}. This representation expresses CVaR as the minimum of an expectation involving a piecewise-linear loss function, which has made mean--CVaR models particularly attractive in both stochastic and robust portfolio optimization.

Expectiles, introduced by \citet{newey1987asymmetric} as asymmetric least-squares generalizations of the mean. \citet{bellini2014generalized}place them within the family of generalized quantiles and establish that they are coherent in the sense of \citet{artzner1999coherent} precisely for confidence levels at or above one half. Expectiles are proposed for risk management by \citet{kuan2009assessing} and are an alternative with a property that CVaR lacks: they are \emph{elicitable}. Elicitability, formalized by \citet{gneiting2011making}, is the existence of a consistent scoring function, and it is what permits a risk forecast to be compared against realized outcomes; \citet{bellini2015elicitable} characterizes which risk measures possess it, and \citet{bellini2019backtesting} constructs the corresponding backtests. Elicitability is what makes a risk forecast testable against realized outcomes, and its absence for CVaR is a recognized practical difficulty. Expectiles also depend on the entire loss distribution rather than on a fixed tail slice, so they respond to the magnitude of losses and not merely to their ordering. \citet{ziegel2016coherence} showed that expectiles are the only law-invariant coherent risk measures that are elicitable. \citet{bellini2017risk} establishes the correspondence between the expectile level and a target exceedance frequency. \citet{jakobsons2016scenario} gives a scenario-aggregation linear program for portfolio expectile minimization, exploiting the dual representation of the expectile as a supremum over a set of densities with bounded ratio. Expectiles have subsequently been examined from the perspective of portfolio performance and reward--risk evaluation \citep{rossello2022performance}. These studies establish expectiles as a tractable and economically meaningful alternative to CVaR when the return distribution is treated as known or is replaced directly by its empirical counterpart. Their performance may nevertheless be sensitive to sampling variation, since the empirical distribution represents only one estimate of the unknown return-generating process.

Distributionally robust optimization addresses uncertainty in the underlying probability distribution by replacing a single estimated distribution with an ambiguity set of plausible distributions. The optimization problem is then solved with respect to the worst-case distribution within this set. Early DRO models commonly employed moment-based ambiguity sets \citep{delage2010distributionally}. Such sets are defined through information on moments such as the mean and covariance but do not fully preserve the shape of the empirical distribution. Wasserstein ambiguity sets provide an alternative by constructing a neighborhood around the empirical distribution and allowing controlled perturbations of the observed data. They have become widely used because they retain distributional information from the sample and admit tractable finite-dimensional reformulations for a broad class of optimization problems. The duality framework of \citet{mohajerin2018data} provides a fundamental basis for such reformulations while further developments are given by \citet{gao2023distributionally} and \citet{zhang2025short}. Statistical support for Wasserstein ambiguity sets is provided by concentration results such as those of \citet{fournier2015rate} and \citet{bolley2007quantitative}.

Applications of Wasserstein distributionally robust optimization to portfolio selection have largely focused on CVaR. The representation of \citet{rockafellar2000optimization} expresses CVaR through an expectation of a piecewise linear loss function and therefore facilitates tractable reformulations of the inner distributional problem. \citet{luan2022robust} study a single-period worst-case mean--CVaR portfolio model. More recently \citet{wang2026multi} extend this framework to a multi-period tri-level model with a target-return constraint imposed within the ambiguity set. Their formulation provides the basic multi-period architecture considered in this paper. An important element of a Wasserstein ambiguity set is the choice of ground metric. This choice determines the geometry of admissible distributional perturbations and also determines the dual norm that appears in the reformulated problem. In portfolio models this dual norm may depend on the allocation vector. The ground metric can therefore affect not only tractability but also the resulting portfolio allocations.

Worst-case expectiles under Wasserstein ambiguity have been studied in a univariate model by \citet{hu2024extreme}. Their analysis uses the supremum representation of expectiles developed by \citet{bellini2014generalized}. This representation expresses an expectile through a one-parameter family of combinations of expected shortfall and the mean. Since this family does not depend on the underlying distribution the two supremum operations can be interchanged. The resulting inner problem is then expressed as a worst-case distortion risk measure using the characterization of \citet{liu2022inf}.

The setting considered here differs in two important respects. First the loss distribution is induced by the portfolio allocation $x_t$. Second a mean-return equality constraint is imposed within the Wasserstein ambiguity set. These features couple the allocation decision with the inner distributional problem and prevent the preceding univariate argument from providing the required portfolio reformulation directly. Section~\ref{sec:properties} discusses this distinction in detail. Remark~\ref{rem:hcm} further shows that Proposition~\ref{prop:degeneracy} reduces to Theorem~2(i) of \citet{hu2024extreme} in the corresponding scalar setting.

This leaves a distinct problem in multi-period portfolio optimization. To the best of our knowledge an exact formulation for a multi-period expectile portfolio problem with a continuous-support Wasserstein ambiguity set and a return constraint imposed within that set has not yet been developed. The present paper addresses this problem.

The analysis also reveals that the behavior of the robust model is governed jointly by the ambiguity radius and the ground metric. In particular the sensitivity of the worst-case expectile to the Wasserstein radius depends on an endogenous exceedance probability rather than solely on an exogenous confidence level. Moreover the model admits a decision-dependent critical radius beyond which the tail component of the expectile reformulation becomes inactive. The resulting portfolio behavior then depends critically on the ground metric. Under the $\ell_{\infty}$ ground metric used in \citet{wang2026multi} the corresponding dual norm is constant on the simplex and the limiting allocation favors concentration. Under the $\ell_{1}$ ground metric adopted in this paper the dual norm is $\ell_{\infty}$ and the limiting allocation instead favors diversification. These results show that both the size and the geometry of the ambiguity set affect the economic behavior of the robust portfolio model.

The radius of the ambiguity set is one of the free parameter. Asymptotic rules follow from the concentration bounds cited above, but these assume independent sampling, whereas financial returns exhibit volatility clustering. We therefore assess data-driven radius calibration with a dependence-preserving stationary bootstrap \citep{politis1994stationary} and evaluate it relative to a model-implied critical radius.
\paragraph{Contributions}
\begin{enumerate}
\item We establish an \emph{envelope theorem} (Theorem~\ref{thm:envelope}) that characterizes the worst-case expectile over any ambiguity set as the unique root of a robust expectation. The argument applies more generally to risk measures with a monotone fixed-point representation.

\item We reformulate the multi-period tri-level problem exactly as four
parametric linear programs with $O(N)$ constraints
(Theorem~\ref{thm:reformulation}). An endpoint result
(Lemma~\ref{lem:betaL}) reduces the number of cases that would arise from a direct enumeration and keeps the formulation linear in the sample size.

\item We derive four structural results: a closed-form, endogenously damped
price of robustness (Proposition~\ref{prop:sensitivity}); a decision-dependent
degeneracy threshold (Proposition~\ref{prop:degeneracy}); exact recovery of the
nominal model at zero radius (Proposition~\ref{prop:nominal}); and a
characterization of how the ground metric governs concentration or
diversification (Proposition~\ref{prop:norm}). The scalar special case of the
threshold is consistent with \citet{hu2024extreme}; its dependence on the
portfolio allocation and the ground metric is new.

\item We provide a rolling empirical study of 90 FTSE constituents over 3,341
out-of-sample trading days. The design matches competing expectile and CVaR
models on the ambiguity set, radius, ground metric, and trade-off weight, and
uses dependence-robust inference. It also documents the predicted
$\ell_\infty$ degeneracy in observed data.
\end{enumerate}

Section~\ref{sec:model} develops the model, Section~\ref{sec:reformulation}
derives the exact reformulation, and Section~\ref{sec:properties} presents the
structural results. Section~\ref{sec:empirical} reports the empirical study, and Section~\ref{sec:conclusions} concludes. Analytical proofs are in~\ref{appen}; the online supplement provides sensitivity analyses, statistical inference, and additional diagnostics.

\section{Model construction}
\label{sec:model}

\subsection{Notations and Definitions}
\label{sec:notation}

We consider an investor who allocates wealth across $n$ risky assets over $T$ consecutive investment periods, rebalancing at the end of each period. Let
$\xi_{t}=[\xi_{t,1},\xi_{t,2},\dots,\xi_{t,n}]^{\top}$ denote the random return vector of the $n$ assets in period $t$, defined on a probability space $(\Omega,\mathcal{F}_{t},P_{t})$ with true distribution $F_{t}$ and support $\Xi\subseteq\mathbb{R}^{n}$. Let $x_{t}=[x_{t,1},\dots,x_{t,n}]^{\top}$ denote the wealth allocation decided at the end of period $t-1$ and held through period $t$. Throughout, the investor is restricted to the simplex
\begin{equation}
\mathcal{X}=\bigl\{x\in\mathbb{R}^{n}\;:\;e^{\top}x=1,\;x\geq 0\bigr\},
\label{eq:simplex}
\end{equation}
where $e$ denotes the vector of ones. The no short-selling constraint requires all portfolio weights to be nonnegative and allows zero allocations to individual assets. The portfolio loss in period $t$ is
\begin{equation}
L_{x_{t}}:=-\xi_{t}^{\top}x_{t}.
\nonumber
\end{equation}
All random variables are assumed integrable, $L_{x_{t}}\in L^{1}(\Omega,\mathcal{F}_{t},P_{t})$; finiteness of the first moment is essential for the expectile risk measure to be well defined
\citep{newey1987asymmetric}.

\begin{definition}[The expectile risk measure; \citealp{newey1987asymmetric}]
For a confidence level $\tau\in(0,1)$, the $\tau$-expectile of the portfolio loss $L_{x}$ under a distribution $F$ is the unique solution $e^{F}_{\tau}(L_{x})$ of
\begin{equation}
\tau\,\mathbb{E}_{F}\bigl[(L_{x}-e^{F}_{\tau}(L_{x}))^{+}\bigr]
-(1-\tau)\,\mathbb{E}_{F}\bigl[(L_{x}-e^{F}_{\tau}(L_{x}))^{-}\bigr]=0,
\label{eq:expectile-def}
\end{equation}
where $(\cdot)^{+}=\max\{0,\cdot\}$ and $(\cdot)^{-}=\max\{0,-\,\cdot\}$. Equivalently, $e^{F}_{\tau}$ is the minimizer of the asymmetric quadratic loss
$\mathbb{E}_{F}[\tau((L_{x}-e)^{+})^{2}+(1-\tau)((L_{x}-e)^{-})^{2}]$.
\label{sec:expectile}

\end{definition}

The expectile satisfies the following properties \citep{bellini2017risk,
bellini2014generalized}: for all $L^{1}_{x},L^{2}_{x}\in L^{1}$,
\begin{enumerate}
\item[(a)] $e_{\tau}(aL_{x}+b)=a\,e_{\tau}(L_{x})+b$ for every $a>0$, $b\in\mathbb{R}$
      (positive homogeneity and cash invariance);
\item[(b)] $L^{1}_{x}\leq L^{2}_{x}$ almost surely implies
      $e_{\tau}(L^{1}_{x})\leq e_{\tau}(L^{2}_{x})$ (monotonicity), for every $\tau\in(0,1)$;
\item[(c)] $e_{\tau}(L^{1}_{x}+L^{2}_{x})\leq e_{\tau}(L^{1}_{x})+e_{\tau}(L^{2}_{x})$
      whenever $\tau\in[\tfrac{1}{2},1)$ (subadditivity), the inequality being reversed for
      $\tau\in(0,\tfrac{1}{2}]$;
\item[(d)] $e_{\tau}(-L_{x})=-e_{1-\tau}(L_{x})$.
\end{enumerate}
Properties (a)--(c) establish that $e_{\tau}$ is a coherent risk measure in the sense of \citet{artzner1999coherent} precisely on the range $\tau\in[\tfrac{1}{2},1)$, which we adopt throughout. 

Throughout the paper we use the constant
\begin{equation}
\kappa_{\tau}:=\frac{2\tau-1}{1-\tau}\;\geq\;0,
\qquad \tau\in[\tfrac{1}{2},1).
\nonumber
\end{equation}
It is strictly increasing in $\tau$, vanishes at $\tau=\tfrac12$ (where the expectile reduces to the mean) and diverges as $\tau\uparrow 1$.

\subsection{Separability and time consistency}
\label{sec:timeconsistency}

Multi-period risk optimization raises the issue of time consistency, namely whether a policy that is optimal at the initial decision date remains optimal at subsequent rebalancing intervals. This issue is well documented for dynamic mean--variance models \citep{lioui2013time,cui2017time}. We follow the approach of \citet{chen2013optimal}, in which a separable conditional risk mapping allows the multi-period problem to be decomposed into a sequence of per-period optimization problems while preserving time consistency. A key requirement for this decomposition is the monotonicity of the per-period risk mapping. Since the expectile functional is monotone, this condition is satisfied in the present framework.

\begin{definition}[Separable conditional risk mapping; \citealp{chen2013optimal}]
A conditional risk mapping $\rho_{t_{1},t_{2}}(\cdot)$ is separable if it decomposes into a sum of single-period components,
\begin{equation}
\rho_{t_{1},t_{2}}\bigl(-w_{t_{1},t_{2}}\bigr)
=\sum_{i=t_{1}+1}^{t_{2}}\mathbb{E}\bigl[\rho_{i}(-w_{i}\mid\mathcal{F}_{i-1})\mid\mathcal{F}_{t_{1}}\bigr],
\label{eq:separable}
\end{equation}
where $-w_{t_{1},t_{2}}$ denotes the investor's loss from $t_{1}$ to $t_{2}$.
\label{def:separable}
\end{definition}

The objective adopted in this paper aggregates per-period mean-expectile risk terms,
\begin{equation}
\min_{x}\;\sum_{t=1}^{T}
\Bigl[(1-\lambda)\,\mathbb{E}_{F_{t}}\bigl(-\xi_{t}^{\top}x_{t}\bigr)
+\lambda\,e^{F_{t}}_{\tau}\bigl(-\xi_{t}^{\top}x_{t}\bigr)\Bigr],
\qquad \tau\in[\tfrac12,1),\;\lambda\in(0,1).
\label{eq:objective}
\end{equation}
Each summand involves only the period-$t$ allocation $x_{t}$ and the period-$t$ return vector $\xi_{t}$, the intertemporal linkage being reduced to the normalization $e^{\top}x_{t}=1$.

\begin{lemma}[Separability and time consistency]
\label{lem:timeconsistency}
For $\tau\in[\tfrac12,1)$ and $\lambda\in(0,1)$, the per-period mapping
\[
\rho_{t}(X_{t+1})=(1-\lambda)\,\mathbb{E}_{F_{t}}\bigl(-\xi_{t}^{\top}x_{t}\bigr)
+\lambda\,e^{F_{t}}_{\tau}\bigl(-\xi_{t}^{\top}x_{t}\bigr)
\]
is monotone and coherent. Consequently the objective \eqref{eq:objective} is separable in the sense of Definition~\ref{def:separable}, the multi-period program decouples across periods, and sequential re-optimization at each rebalancing date is time consistent.
\end{lemma}

Lemma~\ref{lem:timeconsistency} identifies $\tau\geq\tfrac12$ as the substantive restriction: for $\tau<\tfrac12$ subadditivity fails, coherence is lost. We
emphasize that the time-inconsistency documented by \citet{cui2017time} for
non-separable terminal wealth objectives does not arise here, precisely because
\eqref{eq:objective} is a separable sum of per-period functionals rather than a single risk functional applied to terminal wealth.

\subsection{Why the substitution is not a reparametrization}
\label{sec:nonrepresentability}

Both $\mathrm{CVaR}_{\alpha}$ and $e_{\tau}$ are coherent and law invariant, which invites the conjecture that a robust mean-expectile model is a robust mean-CVaR model with different constants. The following lemma rules this out.

\begin{lemma}[Non-representability]
\label{lem:nonrep}
For $\tau\in(\tfrac12,1)$ the expectile $e_{\tau}$ is coherent and law invariant but not comonotone additive. Consequently $e_{\tau}$ is not a spectral risk measure and admits no representation as a single mixture
$e_{\tau}(\cdot)=\int_{0}^{1}\mathrm{CVaR}_{\alpha}(\cdot)\,d\mu(\alpha)$
for any probability measure $\mu$ on $[0,1)$.
\end{lemma}

Lemma~\ref{lem:nonrep} has a concrete methodological consequence. The tractability of Wasserstein distributionally robust mean-CVaR formulations rests on the Rockafellar--Uryasev representation
$\mathrm{CVaR}_{\alpha}(L)=\min_{q}\{q+\tfrac{1}{1-\alpha}\mathbb{E}(L-q)^{+}\}$
\citep{rockafellar2000optimization}, which renders CVaR an expectation of a piecewise linear integrand and therefore places the inner worst-case problem directly within the scope of Wasserstein duality \citep{mohajerin2018data,gao2023distributionally}. The expectile admits no
such representation: it is defined implicitly through \eqref{eq:expectile-def}, and its dual representation
\begin{equation}
e^{F}_{\tau}(L_{x})=\sup_{\varphi\in\Phi_{\tau}(F)}\;\mathbb{E}_{F}(\varphi L_{x}),
\qquad
\Phi_{\tau}(F)=\left\{\begin{aligned}
&\varphi\in L^{\infty}:\ \varphi\geq0,\ \mathbb{E}_{F}(\varphi)=1,\\
&(1-\tau)m\leq\varphi\leq\tau m\ \text{for some }m>0
\end{aligned}\right\},
\label{eq:expectile-dual}
\end{equation}
ranges over a feasible set that depends on the distribution $F$ itself, through both the normalization $\mathbb{E}_{F}(\varphi)=1$ and the free scale $m$ \citep{jakobsons2016scenario}. Under ambiguity, \eqref{eq:expectile-dual} yields a nested supremum $\sup_{F}\sup_{\varphi\in\Phi(F)}$ whose two levels are coupled. Resolving this coupling is the technical contribution of Section~\ref{sec:reformulation}.

\subsection{Wasserstein ambiguity sets and radius calibration}
\label{sec:ambiguity}

The true return distribution $F_t$ is unknown and must be inferred from a finite sample. Consequently, portfolio allocations based on a single estimated distribution may be sensitive to distributional misspecification. To account for this uncertainty, we consider a Wasserstein ambiguity set centered at the empirical distribution, following \citet{mohajerin2018data}.

\begin{definition}[Wasserstein metric; \citealp{villani2009optimal}]
On $M(\Xi)\times M(\Xi)$ the order-$p$ Wasserstein distance is
\begin{equation}
W_{p}(F_{1},F_{2})=\Bigl(\inf_{\pi\in M(\Xi^{2})}
\int_{\Xi^{2}}\norm{\xi_{1}-\xi_{2}}^{p}\,d\pi(\xi_{1},\xi_{2})\Bigr)^{1/p},
\label{eq:wasserstein}
\end{equation}
where $F_{1},F_{2}\in M(\Xi)$ are the marginals of $\xi_{1},\xi_{2}$ and $\pi$ is a coupling.
\label{def:wasserstein}
\end{definition}

We set $p=1$ throughout, which is what permits the exact reformulation of
Section~\ref{sec:reformulation}. The ambiguity set for period $t$ is
\begin{equation}
\mathcal{P}_{t}(F_{t,N_{t}})=\bigl\{F_{t}\in M(\Xi)\;:\;W_{1}(F_{t},F_{t,N_{t}})\leq\rho_{t}\bigr\},
\label{eq:ambiguityset}
\end{equation}
centered at the empirical measure
\begin{equation}
F_{t,N_{t}}=\frac{1}{N_{t}}\sum_{i=1}^{N_{t}}\delta_{\hat{\xi}_{t,i}},
\nonumber
\end{equation}
where $\delta$ denotes the Dirac measure, $\hat{\xi}_{t,i}$ the $i$th sampled return vector and $N_{t}$ the sample size. We impose the standing moment condition
\begin{equation}
\int_{\Xi}\norm{\xi_{t}}\,F_{t}(d\xi_{t})<\infty,
\label{eq:moment}
\end{equation}
under which the Wasserstein ball contains the true distribution with high probability for an appropriately chosen radius \citep{fournier2015rate,bolley2007quantitative}.

\begin{remark}
Under an $\ell_{\infty}$ ground metric, the associated dual norm is $\ell_{1}$. Since $\norm{x_{t}}_{1}\equiv1$ on the simplex \eqref{eq:simplex}, the term $\dnorm{x_t}$ becomes constant and the robust adjustment does not distinguish among feasible allocations. In contrast, under an $\ell_{1}$ ground metric, the associated dual norm is $\ell_{\infty}$ and $\norm{x_t}_{\infty}$ remains decision dependent. Propositions~\ref{prop:degeneracy} and~\ref{prop:norm} make the consequences precise, and Section~\ref{sec:h2} tests them.
\end{remark}

\subsection{The multi-period distributionally robust mean-expectile model}
\label{sec:trilevel}

Combining \eqref{eq:objective} with the ambiguity set \eqref{eq:ambiguityset} and a per-period minimum return requirement $r_{t}$ yields the model
\begin{align}
\min_{x}\;&\sum_{t=1}^{T}\left\{\max_{\beta_{t}\geq r_{t}}\;
\max_{F_{t}\in\mathcal{P}_{t}(F_{t,N_{t}})}
\Bigl[(1-\lambda)\,\mathbb{E}_{F_{t}}\bigl(-\xi_{t}^{\top}x_{t}\bigr)
+\lambda\,e^{F_{t}}_{\tau}\bigl(-\xi_{t}^{\top}x_{t}\bigr)\Bigr]\right\}
\label{eq:model-a}\\[2pt]
\text{s.t. }\;&\beta_{t}\geq r_{t}, &&t=1,\dots,T,
\label{eq:model-b}\\
&x_{t}\in\mathcal{X}, &&t=1,\dots,T,
\label{eq:model-c}
\end{align}
where the equality $\mathbb{E}_{F_{t}}(\xi_{t}^{\top}x_{t})=\beta_{t}$ is imposed within $\mathcal{P}_{t}(F_{t,N_{t}})$.

The outer minimization over $x_t$ represents the investor's portfolio allocation decision. The intermediate maximization over $\beta_t$ captures the requirement that the expected portfolio return should remain above the prescribed return target $r_t$. The adversary therefore selects the most unfavorable admissible value of $\beta_t$ consistent with the ambiguity set. The inner maximization over $F_{t}$ is the worst-case return distribution within the Wasserstein ball.

\begin{remark}
\label{rem:beta}
Since $e^{F}_{\tau}$ is not
an expectation, the supremum of the sum in \eqref{eq:model-a} does not in general equal the sum of the suprema. With the mean pinned at $-\beta_{t}$ on
$\mathcal{P}_{t}(\beta_{t}):=\mathcal{P}_{t}(F_{t,N_{t}})\cap\{F:\mathbb{E}_{F}(\xi_{t}^{\top}x_{t})=\beta_{t}\}$, however, the decomposition is exact:
\begin{equation}
\max_{F_{t}\in\mathcal{P}_{t}(\beta_{t})}
\Bigl[(1-\lambda)(-\beta_{t})+\lambda e^{F_{t}}_{\tau}\Bigr]
=-(1-\lambda)\beta_{t}+\lambda\max_{F_{t}\in\mathcal{P}_{t}(\beta_{t})}e^{F_{t}}_{\tau}
\bigl(-\xi_{t}^{\top}x_{t}\bigr).
\label{eq:split}
\end{equation}
The remainder of the derivation rests on \eqref{eq:split}.
\end{remark}

\begin{remark}[Limitations of the supremum-form representation]
\label{rem:supsup}
The supremum-form representation does not directly resolve the ambiguity problem for two reasons. First, the outer level of the model in Section~\ref{sec:trilevel} involves minimization with respect to $x_t$. Interchanging the relevant operators therefore leads to a $\min_x\max_{\gamma}$ formulation. This results in a semi-infinite optimization problem over a continuum of values of $\gamma$. Its solution generally requires either discretization or a cutting-plane procedure. Consequently, the resulting formulation no longer preserves the exact linear programming representation established in Theorem~\ref{thm:reformulation}. Second, the inner problem includes a mean equality constraint within the ambiguity ball as described in Section~\ref{sec:trilevel}. The closed-form characterization of worst-case used for the interchange in \citet{liu2022inf} is not directly applicable in the presence of this additional constraint. The envelope formulation developed in Section~\ref{sec:envelope} addresses these limitations by introducing a single scalar variable together with one binding constraint that characterizes its value.
\end{remark}

\section{Equivalent reformulation and solution design}
\label{sec:reformulation}
This section reformulates the tri--level program \eqref{eq:model-a}--\eqref{eq:model-c} to a finite family of programs solvable by standard solvers. The reformulation proceeds in four steps: a fixed-point characterization of the expectile (Section~\ref{sec:fixedpoint}); an envelope theorem converting the worst-case expectile into the root of a Wasserstein robust expectation problem (Section~\ref{sec:envelope}); Lagrangian dual of that expectation problem
(Section~\ref{sec:dual}); and resolution of the intermediate maximization
(Sections~\ref{sec:beta}--\ref{sec:theorem2}).

\subsection{A fixed-point characterization of the expectile}
\label{sec:fixedpoint}

\begin{lemma}[Fixed-point representation]
\label{lem:fixedpoint}
Let $\tau\in[\tfrac12,1)$ and $\kappa_{\tau}=(2\tau-1)/(1-\tau)$. For any $F$ with
$\mathbb{E}_{F}|L|<\infty$,
\begin{equation}
e^{F}_{\tau}(L)=\mathbb{E}_{F}[L]+\kappa_{\tau}\,\mathbb{E}_{F}\bigl[(L-e^{F}_{\tau}(L))^{+}\bigr].
\label{eq:fixedpoint}
\end{equation}
\end{lemma}

Lemma~\ref{lem:fixedpoint} converts the implicit definition \eqref{eq:expectile-def} into a scalar fixed-point equation whose right-hand side is an expectation of a piecewise linear function of the loss. Define
\begin{equation}
g(q;F):=\mathbb{E}_{F}[L]+\kappa_{\tau}\,\mathbb{E}_{F}(L-q)^{+}-q,\qquad q\in\mathbb{R}.
\nonumber
\end{equation}
Then $\partial g/\partial q=-\kappa_{\tau}\mathbb{P}_{F}(L>q)-1\leq-1<0$, so $g(\cdot;F)$ is strictly decreasing and $(1+\kappa_{\tau})$-Lipschitz, with unique root $e^{F}_{\tau}(L)$ by Lemma~\ref{lem:fixedpoint}.

\subsection{The worst-case expectile envelope theorem}
\label{sec:envelope}

\begin{theorem}[Worst-case expectile envelope]
\label{thm:envelope}
Let $\mathcal{P}$ be a set of probability measures on $\Xi$ with
$\sup_{F\in\mathcal{P}}\mathbb{E}_{F}\norm{\xi}<\infty$, and let $\tau\in(\tfrac12,1)$. Define
\begin{equation}
\Psi(q):=\sup_{F\in\mathcal{P}}\Bigl[\mathbb{E}_{F}[L]+\kappa_{\tau}\,\mathbb{E}_{F}(L-q)^{+}\Bigr]-q .
\label{eq:Psi}
\end{equation}
Then $\Psi$ is continuous and strictly decreasing on $\mathbb{R}$, possesses a unique root $q^{\star}$, and
\begin{equation}
\sup_{F\in\mathcal{P}}e^{F}_{\tau}(L)=q^{\star}=\min\bigl\{q\in\mathbb{R}:\Psi(q)\leq0\bigr\}.
\nonumber
\end{equation}
\end{theorem}

The detailed proof is provided in \ref{app:envelope}.

Theorem~\ref{thm:envelope} is the key step. It replaces the nested supremum
$\sup_{F}\sup_{\varphi\in\Phi(F)}$ arising from \eqref{eq:expectile-dual} by a single supremum over $F$ of an expectation, evaluated at a scalar fixed point. The following lemma shows that the resulting integrand has a two-piece structure analogous to that derived in \citet{wang2026multi}.

\begin{lemma}[Two-piece structure of the integrand]
\label{lem:twopiece}
For fixed $x_{t}\in\mathcal{X}$ and $q_{t}\in\mathbb{R}$, the integrand appearing in \eqref{eq:Psi} is
\begin{equation}
h_{t}(\xi_{t}):=-\xi_{t}^{\top}x_{t}+\kappa_{\tau}\bigl(-\xi_{t}^{\top}x_{t}-q_{t}\bigr)^{+}
=\max\bigl\{-k_{1}\xi_{t}^{\top}x_{t}+c_{1}q_{t},\;
-k_{2}\xi_{t}^{\top}x_{t}+c_{2}q_{t}\bigr\},
\label{eq:twopiece}
\end{equation}
with
\begin{equation}
k_{1}=1,\qquad k_{2}=1+\kappa_{\tau},\qquad c_{1}=0,\qquad c_{2}=-\kappa_{\tau},\qquad
\Delta:=k_{2}-k_{1}=\kappa_{\tau} .
\nonumber
\end{equation}
In particular $k_{2}>k_{1}\geq0$ for every $\tau\in(\tfrac12,1)$.
\end{lemma}

\begin{remark} In the mean-CVaR inner problem of \citet{wang2026multi}, $k_{1}=1-\lambda$, $k_{2}=(1-\alpha+\lambda\alpha)/(1-\alpha)$ and $\Delta=\lambda/(1-\alpha)$. In contrast, the slope gap $\Delta$ in the present formulation is the endogenous constant $\kappa_{\tau}$ determined by the expectile level rather than the exogenous tail-inflation factor $1/(1-\alpha)$. Proposition~\ref{prop:sensitivity} shows that this distinction also affects the marginal cost of robustness.

\end{remark}

\subsection{Lagrangian dual of the inner maximization}
\label{sec:dual}
 
Fix $t$, $x_{t}$, $q_{t}$ and $\beta_{t}$. From Theorem~\ref{thm:envelope} and
Lemma~\ref{lem:twopiece}, the inner maximization problem can be written as
\begin{align}
\sup_{F_{t}}\;&\int_{\Xi}h_{t}(\xi_{t})\,F_{t}(d\xi_{t})
\label{eq:inner-a}\\
\text{s.t. }&\int_{\Xi}\xi_{t}^{\top}x_{t}\,F_{t}(d\xi_{t})=\beta_{t},
\label{eq:inner-b}\\
&W_{1}\bigl(F_{t},F_{t,N_{t}}\bigr)\leq\rho_{t}.
\label{eq:inner-c}
\end{align}
Introducing multipliers $\gamma_{t,1}\in\mathbb{R}$ for \eqref{eq:inner-b} and
$\gamma_{t,2}\geq0$ for \eqref{eq:inner-c}, and applying the dual-norm identity
$\dnorm{y}=\sup_{\norm{z}\leq1}z^{\top}y$ together with the minimax theorem, standard Wasserstein duality \citep{mohajerin2018data,gao2023distributionally,zhang2025short} yields
\begin{align}
D_{t}(q_{t},\beta_{t})=\min_{\gamma_{t,1},\gamma_{t,2},s}\;&
\gamma_{t,1}\beta_{t}+\gamma_{t,2}\rho_{t}+\frac{1}{N_{t}}\sum_{i=1}^{N_{t}}s_{t,i}
\label{eq:dual-a}\\
\text{s.t. }&s_{t,i}\geq c_{1}q_{t}-(k_{1}+\gamma_{t,1})\hat{\xi}_{t,i}^{\top}x_{t},
&&i=1,\dots,N_{t},
\label{eq:dual-b}\\
&s_{t,i}\geq c_{2}q_{t}-(k_{2}+\gamma_{t,1})\hat{\xi}_{t,i}^{\top}x_{t},
&&i=1,\dots,N_{t},
\label{eq:dual-c}\\
&\dnorm{(k_{1}+\gamma_{t,1})x_{t}}\leq\gamma_{t,2},\qquad
\dnorm{(k_{2}+\gamma_{t,1})x_{t}}\leq\gamma_{t,2},\qquad\gamma_{t,2}\geq0,
\label{eq:dual-d}
\end{align}
with $(k_{1},k_{2},c_{1},c_{2})=(1,1+\kappa_{\tau},0,-\kappa_{\tau})$ from Lemma~\ref{lem:twopiece} and
$\hat{\mu}_{t}=N_{t}^{-1}\sum_{i}\hat{\xi}_{t,i}$ the sample mean return vector.
 
Since $\xi\mapsto\xi^{\top}x_{t}$ is Lipschitz with modulus $\dnorm{x_{t}}$, Kantorovich duality gives
$|\mathbb{E}_{F_{t}}(\xi_{t}^{\top}x_{t})-\hat{\mu}_{t}^{\top}x_{t}|\leq\rho_{t}\dnorm{x_{t}}$.
Hence the constraints \eqref{eq:inner-b}--\eqref{eq:inner-c} are jointly feasible if and only if $\beta_{t}\in[\hat{\mu}_{t}^{\top}x_{t}-\rho_{t}\dnorm{x_{t}},\,
\hat{\mu}_{t}^{\top}x_{t}+\rho_{t}\dnorm{x_{t}}]$. For $\beta_t$ in the interior of this interval the Slater condition holds and therefore strong duality applies.
 
By the choice of ground metric made in Section~\ref{sec:ambiguity} the dual norm is $\dnorm{\cdot}=\norm{\cdot}_{\infty}$, so $\dnorm{x_{t}}$ is decision dependent and does not reduce to a constant on the simplex.
 
Define $w_{t}:=\norm{x_{t}}_{\infty}$ and set
$a_{t,1}:=k_{1}+\gamma_{t,1}=1+\gamma_{t,1},$
conditions \eqref{eq:dual-d} reduces to $|a_{t,1}|\,w_{t}\leq\gamma_{t,2}$ and
$|a_{t,1}+\kappa_{\tau}|\,w_{t}\leq\gamma_{t,2}.$ 

The dual problem can therefore be expressed as
\begin{align}
D_{t}(q_{t},\beta_{t})=\min_{a_{t,1},\gamma_{t,2},s}\;&
(a_{t,1}-1)\beta_{t}+\gamma_{t,2}\rho_{t}+\frac{1}{N_{t}}\sum_{i=1}^{N_{t}}s_{t,i}
\label{eq:dual2-a}\\
\text{s.t. }&s_{t,i}\geq-a_{t,1}\hat{\xi}_{t,i}^{\top}x_{t},
&&i=1,\dots,N_{t},
\label{eq:dual2-b}\\
&s_{t,i}\geq-\kappa_{\tau} q_{t}-(a_{t,1}+\kappa_{\tau})\hat{\xi}_{t,i}^{\top}x_{t},
&&i=1,\dots,N_{t},
\label{eq:dual2-c}\\
&|a_{t,1}|\,w_{t}\leq\gamma_{t,2},\qquad|a_{t,1}+\kappa_{\tau}|\,w_{t}\leq\gamma_{t,2}.
\label{eq:dual2-d}
\end{align}
Because $\rho_{t}>0$ the multiplier binds, and because $w_{t}>0$ it factorises out of
\eqref{eq:dual2-d}:
\begin{equation}
\gamma_{t,2}^{\star}=\gamma_{t,2}^{0}(a_{t,1})\,w_{t},
\qquad
\gamma_{t,2}^{0}(a):=\max\bigl\{|a|,|a+\kappa_{\tau}|\bigr\}
=\begin{cases}
a+\kappa_{\tau}, & a\geq-\kappa_{\tau}/2,\\[2pt]
-a,    & a<-\kappa_{\tau}/2.
\end{cases}
\label{eq:gamma2}
\end{equation}
 
\begin{remark}
\label{rem:inherit}
The function $\gamma_{t,2}^{0}$ in \eqref{eq:gamma2} will coincide with that arising under an $\ell_{\infty}$ ground metric. The change in the ground metric affects the expression only through the factor $w_{t}$, thus analysis of Section~\ref{sec:beta} is unaffected by the choice of ground metric.
\end{remark}
 
\subsection{Resolution of the intermediate maximization}
\label{sec:beta}
 
By Theorem~\ref{thm:envelope}, the worst-case expectile is the minimum $q_{t}$
satisfying $D_{t}(q_{t},\beta_{t})\leq q_{t}$ with equality at the optimum. The period-$t$ value is $V_{t}(\beta_{t})=-(1-\lambda)\beta_{t}+\lambda q_{t}(\beta_{t})$, maximized over
\begin{equation}
\beta_{t}\in[\beta_{t,L},\beta_{t,U}],
\qquad
\beta_{t,L}=\max\bigl\{r_{t},\,\hat{\mu}_{t}^{\top}x_{t}-\rho_{t}w_{t}\bigr\},
\qquad
\beta_{t,U}=\hat{\mu}_{t}^{\top}x_{t}+\rho_{t}w_{t}.
\label{eq:betabounds}
\end{equation}
 
\begin{remark}[The dependence on $\beta_{t}$ is not affine]
\label{rem:slope}
Differentiating $q=D_{t}(q,\beta_{t})$ implicitly, and using the envelope identity $\partial D_{t}/\partial\beta=\gamma_{t,1}
=a_{t,1}-1$ together with $\partial D_{t}/\partial q=-\kappa_{\tau} P$, where
$P:=\mathbb{P}_{F^{\star}}(L_{x_{t}}>q_{t})$ is evaluated under the worst-case measure, gives
\begin{equation}
\frac{dq_{t}}{d\beta_{t}}=\frac{a_{t,1}-1}{1+\kappa_{\tau} P},
\qquad
\frac{dV_{t}}{d\beta_{t}}=-(1-\lambda)+\frac{\lambda\,(a_{t,1}-1)}{1+\kappa_{\tau} P}.
\label{eq:slope}
\end{equation}
The affine expression is recovered only in the degenerate case $\kappa_{\tau} P=0$.
\end{remark}

\begin{lemma}[Sign of the multiplier and selection of $\beta_{t}$]
\label{lem:betaL}
The function $D_{t}(q,\cdot)$ is concave and non-increasing on
$[\beta_{t,L},\beta_{t,U}]$. Consequently $\gamma_{t,1}=\partial D_{t}/\partial\beta\leq0$, so $a_{t,1}\leq1$, and by \eqref{eq:slope}
\begin{equation}
\frac{dV_{t}}{d\beta_{t}}\ \leq\ -(1-\lambda)\ <\ 0
\qquad\text{for every admissible }a_{t,1}.
\label{eq:slope-negative}
\end{equation}
The intermediate maximization is therefore attained at $\beta_{t}^{\star}=\beta_{t,L}$ for all parameter values.
\end{lemma}

By Lemma~\ref{lem:betaL} the upper endpoint $\beta_{t,U}$ is never selected, so the branches arising from it do not occur. What remains is the kink of $\gamma_{t,2}^{0}$ at $a_{t,1}=-\kappa_{\tau}/2$, crossed with the two
configurations of $\beta_{t,L}$ in \eqref{eq:betabounds},
\begin{align}
\text{Case A:}\quad &\hat{\mu}_{t}^{\top}x_{t}-\rho_{t}w_{t}\geq r_{t}
&&\Longrightarrow\quad\beta_{t,L}=\hat{\mu}_{t}^{\top}x_{t}-\rho_{t}w_{t},
\label{eq:caseA}\\
\text{Case B:}\quad &\hat{\mu}_{t}^{\top}x_{t}-\rho_{t}w_{t}<r_{t}
\leq\hat{\mu}_{t}^{\top}x_{t}+\rho_{t}w_{t}
&&\Longrightarrow\quad\beta_{t,L}=r_{t},
\label{eq:caseB}
\end{align}
giving four mutually exclusive and collectively exhaustive branches. Substituting
\eqref{eq:caseA}--\eqref{eq:caseB} and \eqref{eq:gamma2} into the binding fixed point:
\begin{align}
\text{Case A:}\quad
q_{t}&\geq(a_{t,1}-1)\hat{\mu}_{t}^{\top}x_{t}
+\bigl[\gamma_{t,2}^{0}(a_{t,1})-(a_{t,1}-1)\bigr]\rho_{t}w_{t}
+\frac{1}{N_{t}}\sum_{i=1}^{N_{t}}s_{t,i},
\label{eq:R1A}\\
\text{Case B:}\quad
q_{t}&\geq(a_{t,1}-1)r_{t}+\gamma_{t,2}^{0}(a_{t,1})\rho_{t}w_{t}
+\frac{1}{N_{t}}\sum_{i=1}^{N_{t}}s_{t,i}.
\label{eq:R1B}
\end{align}
The coefficient on $w_{t}$ in \eqref{eq:R1A} is strictly positive, since
$\gamma_{t,2}^{0}(a)\geq|a|\geq a>a-1$.

\subsection{Equivalent reformulation}
\label{sec:theorem2}
 
\begin{theorem}[Equivalent reformulation]
\label{thm:reformulation}
Let $\tau\in(\tfrac12,1)$, $\lambda\in(0,1)$, $\rho_{t}>0$, and let $F_{t}\in M(\Xi)$ satisfy \eqref{eq:moment}. The optimal value of Model
\eqref{eq:model-a}--\eqref{eq:model-c} equals the minimum of the optimal values of the four programs \emph{EX-A1}, \emph{EX-A2}, \emph{EX-B1} and \emph{EX-B2} stated below, and the corresponding multi-period optimal allocations are obtained from their solutions. Each program is linear in $(x_{t},q_{t},s_{t,i},w_{t})$ for fixed $a_{t,1}$.
\end{theorem}
 
\paragraph{Case A ($\hat{\mu}_{t}^{\top}x_{t}-\rho_{t}w_{t}\geq r_{t}$)}
 
\emph{EX-A1} ($-\kappa_{\tau}/2\leq a_{t,1}\leq1$, hence $\gamma_{t,2}^{0}=a_{t,1}+\kappa_{\tau}$):
\begin{align}
\min_{x_{t},q_{t},a_{t,1},s_{t,i},w_{t}}\;
&\sum_{t=1}^{T}\Bigl[-(1-\lambda)\bigl(\hat{\mu}_{t}^{\top}x_{t}-\rho_{t}w_{t}\bigr)
+\lambda q_{t}\Bigr]
\label{eq:EXA1-a}\\
\text{s.t. }
&q_{t}\geq(a_{t,1}-1)\hat{\mu}_{t}^{\top}x_{t}
+\bigl(\kappa_{\tau}+1\bigr)\rho_{t}w_{t}
+\frac{1}{N_{t}}\sum_{i=1}^{N_{t}}s_{t,i},
\label{eq:EXA1-b}\\
&s_{t,i}\geq-a_{t,1}\hat{\xi}_{t,i}^{\top}x_{t},
&&\hspace{-2.4cm}i=1,\dots,N_{t},
\label{eq:EXA1-c}\\
&s_{t,i}\geq-\kappa_{\tau} q_{t}-(a_{t,1}+\kappa_{\tau})\hat{\xi}_{t,i}^{\top}x_{t},
&&\hspace{-2.4cm}i=1,\dots,N_{t},
\label{eq:EXA1-d}\\
&\hat{\mu}_{t}^{\top}x_{t}-\rho_{t}w_{t}\geq r_{t},
\label{eq:EXA1-e}\\
&w_{t}\geq x_{t,j},
&&\hspace{-2.4cm}j=1,\dots,n,
\label{eq:EXA1-f}\\
&-\kappa_{\tau}/2\leq a_{t,1}\leq1,
\label{eq:EXA1-g}\\
&x_{t}\geq0,\qquad e^{\top}x_{t}=1,\qquad w_{t}\geq0.
\label{eq:EXA1-h}
\end{align}
The coefficient $(\kappa_{\tau}+1)$ in \eqref{eq:EXA1-b} is
$\gamma_{t,2}^{0}(a_{t,1})-(a_{t,1}-1)=(a_{t,1}+\kappa_{\tau})-(a_{t,1}-1)$, which is independent of $a_{t,1}$ on this regime.
 
\emph{EX-A2} ($a_{t,1}<-\kappa_{\tau}/2$, hence $\gamma_{t,2}^{0}=-a_{t,1}$): as \emph{EX-A1} with
\eqref{eq:EXA1-b} replaced by
\begin{equation}
q_{t}\geq(a_{t,1}-1)\hat{\mu}_{t}^{\top}x_{t}
+\bigl(1-2a_{t,1}\bigr)\rho_{t}w_{t}
+\frac{1}{N_{t}}\sum_{i=1}^{N_{t}}s_{t,i},
\nonumber
\end{equation}
and \eqref{eq:EXA1-g} replaced by $a_{t,1}<-\kappa_{\tau}/2$; constraints
\eqref{eq:EXA1-a}, \eqref{eq:EXA1-c}--\eqref{eq:EXA1-f} and \eqref{eq:EXA1-h} are
unchanged.
 
\paragraph{Case B ($\hat{\mu}_{t}^{\top}x_{t}-\rho_{t}w_{t}<r_{t}
\leq\hat{\mu}_{t}^{\top}x_{t}+\rho_{t}w_{t}$)}
 
\emph{EX-B1} ($-\kappa_{\tau}/2\leq a_{t,1}\leq1$) and \emph{EX-B2} ($a_{t,1}<-\kappa_{\tau}/2$) are
obtained from \emph{EX-A1} and \emph{EX-A2} respectively by replacing the objective
\eqref{eq:EXA1-a} with
\begin{equation}
\min\;\sum_{t=1}^{T}\bigl[-(1-\lambda)r_{t}+\lambda q_{t}\bigr],
\nonumber
\end{equation}
the fixed-point constraint with
\begin{equation}
q_{t}\geq(a_{t,1}-1)r_{t}+\gamma_{t,2}^{0}(a_{t,1})\rho_{t}w_{t}
+\frac{1}{N_{t}}\sum_{i=1}^{N_{t}}s_{t,i},
\label{eq:EXB-b}
\end{equation}
and the case condition \eqref{eq:EXA1-e} with the pair
\begin{equation}
\hat{\mu}_{t}^{\top}x_{t}-\rho_{t}w_{t}\leq r_{t},
\qquad
\hat{\mu}_{t}^{\top}x_{t}+\rho_{t}w_{t}\geq r_{t}.
\label{eq:EXB-c}
\end{equation}
 
The proof of Theorem~\ref{thm:reformulation} is given in
Appendix~B; it verifies that \eqref{eq:gamma2}, \eqref{eq:caseA}--\eqref{eq:caseB}
exhaust the parameter space and that the substitutions of $\gamma_{t,2}^{\star}$ and
$\beta_{t}^{\star}$ are exact within each branch.
 
\begin{remark}[Exactness of the linearization of $\|x_t\|_\infty$]\label{rem:w}
The constraints \eqref{eq:EXA1-f} impose only $w_t\ge x_{t,i}$ and leave $w_t$ free to exceed $\|x_t\|_\infty$. In Case~A this is immaterial: the objective \eqref{eq:EXA1-a}, the constraints \eqref{eq:EXA1-b} and \eqref{eq:EXA1-e} penalize larger $w_t$, so the constraints bind and $w_t^\star=\|x_t\|_\infty$. In Case~B the constraints \eqref{eq:EXB-c} increases $w_t$
and
\(
w_t^\star=\max\Bigl\{\|x_t\|_\infty,\ \bigl|\hat{\mu}_t^{\top}x_t-r_t\bigr|/\rho_t\Bigr\},
\)
with equality to $\|x_t\|_\infty$ precisely on the region where the Case~B condition \eqref{eq:caseB} holds. Outside that region the large $w_t$ increases $q_t$ through \eqref{eq:EXB-b} and so penalizes the allocation. Each branch is thus exact on its own case region and conservative elsewhere, which is all the minimization over the four branches requires.
\end{remark}

\begin{remark}[Fixed point without bisection]\label{rem:no-bisection}
The constraints \eqref{eq:EXA1-b} and \eqref{eq:EXB-b} represents the fixed point of Theorem~\ref{thm:envelope} as inequalities rather than equalities. This is exact because $q_t$ enters the objective with the strictly positive coefficient $\lambda$. Hence, the minimization selects the smallest feasible value of $q_t$, at which the corresponding constraint holds with equality. This value coincides with the fixed point $q_t^\star$ and therefore no outer bisection procedure is required. Since $q_t$ appears in
\eqref{eq:EXA1-d} only with coefficient $-\kappa_\tau$, the system remains linear in $(q_t,s_{t,i})$.
\end{remark}

\begin{remark}[The trade-off weight is inoperative in Case~B]\label{rem:caseB}
In Case~B the selected endpoint $\beta_{t,L}=r_t$ does not depend on $x_t$, so $-(1-\lambda)r_t$ is an additive constant and $\lambda$ leaves the argmin: the program reduces to worst-case expectile minimization subject to \eqref{eq:EXB-c}. The share of rebalancing dates in Case~A therefore determines whether the mean--risk trade-off was active at all and should be reported alongside performance. Under the $\ell_1$ ground metric the Case~A condition \eqref{eq:caseA} is easy to satisfy for diversified portfolio than $\ell_\infty$ ground metric, since
$\|x_t\|_\infty\le\|x_t\|_1=1$. Hence, the mean--risk trade-off remains relevant over a larger region of the feasible set.
\end{remark}

\subsection{Solution design and computational complexity}

The programs of Theorem~\ref{thm:reformulation} are linear in $(x_t,q_t,s_{t,i},w_t)$ only for a fixed multiplier: the products $a_{t,1}\hat{\mu}_t^{\top}x_t$ in \eqref{eq:EXA1-b} and
$a_{t,1}\hat{\xi}_{t,i}^{\top}x_t$ in \eqref{eq:EXA1-c}--\eqref{eq:EXA1-d} are bilinear in the dual multiplier and the allocation vector. We treat this bilinearity exactly, and two features of the reformulation permit it. First, once the intermediate maximization has been resolved as in Section~\ref{sec:reformulation}, $a_{t,1}$ is no longer adversarial: what remains is a joint minimization over $(x_t,a_{t,1})$, so the order of optimization may be interchanged without loss.
Second, by Lemma~\ref{lem:timeconsistency} the periods decouple, so $a_{t,1}$ is a single scalar per rebalancing date, by Lemma~\ref{lem:betaL} $a_{t,1}\le 1$ and bounded below within each branch.

The period-$t$ subproblem is therefore a one-dimensional search over a bounded interval, at each point of which a linear program is solved. Within each of the four branches we restrict $a_{t,1}$ to its regime interval---$[-\kappa_\tau/2,\,1]$ for EX-A1 and EX-B1, $[-a_{\max},\,-\kappa_\tau/2)$
for EX-A2 and EX-B2---evaluate a coarse grid, refine locally about the incumbent, and return the allocation attaining the smallest objective value across the four branches. The coarse grid is needed because the value function $a\mapsto\min_{x_t,q_t,s,w_t}\{\text{branch objective}\}$ is not in general unimodal, so a pure golden-section search may terminate at a local minimizer.

Each linear program contains $n+N_t+2$ variables, consisting of the allocation vector $x_t$, the worst-case expectile $q_t$, the $N_t$ epigraph variables $s_{t,i}$ and the auxiliary variable $w_t$. The formulation has $2N_t+n+3$ constraints in Case~A and $2N_t+n+4$ constraints in Case~B, with the additional constraint in Case~B arising from the second inequality in \eqref{eq:EXB-c}.

\section{Structural properties}
\label{sec:properties}

The reformulation yields four results of independent interest. To isolate the effect of the ambiguity radius, Propositions~\ref{prop:sensitivity} and~\ref{prop:degeneracy} are stated for the ambiguity set \eqref{eq:ambiguityset} at fixed $x_{t}$, that is without the mean constraint \eqref{eq:inner-b}; the model inherits the qualitative conclusions. Let
\begin{equation}
\hat{L}_{t,i}:=-\hat{\xi}_{t,i}^{\top}x_{t},
\qquad
\bar{L}_{t}:=\frac{1}{N_{t}}\sum_{i=1}^{N_{t}}\hat{L}_{t,i}=-\hat{\mu}_{t}^{\top}x_{t},
\qquad
\overline{(\hat{L}_{t}-q)^{+}}:=\frac{1}{N_{t}}\sum_{i=1}^{N_{t}}(\hat{L}_{t,i}-q)^{+},
\nonumber
\end{equation}
for the realized sample losses, their mean and the mean sample exceedance above a level
$q\in\mathbb{R}$. Let
\begin{equation}
\hat{P}_{t}(q):=\frac{1}{N_{t}}\bigl|\{i:\hat{L}_{t,i}>q\}\bigr|
\qquad\text{and}\qquad
\mathrm{sp}_{t}(x_{t}):=\max_{1\leq i\leq N_{t}}\hat{L}_{t,i}-\bar{L}_{t}
=\max_{1\leq i\leq N_{t}}(\hat{\mu}_{t}-\hat{\xi}_{t,i})^{\top}x_{t}
\nonumber
\end{equation}
denote the empirical exceedance frequency and the empirical upper spread of the portfolio loss. The frequency $\hat{P}_{t}$ is evaluated under $F_{t,N_{t}}$ and is to be distinguished from $P=\mathbb{P}_{F^{\star}}(L_{x_{t}}>q_{t})$ of Remark~\ref{rem:slope}, which is evaluated under the worst-case measure. Being a maximum of linear functions of $x_{t}$, the map $\mathrm{sp}_{t}$ is sublinear, that is convex and positively homogeneous; and since $\sum_{i}(\hat{\mu}_{t}-\hat{\xi}_{t,i})^{\top}x_{t}=0$ we have
$\mathrm{sp}_{t}(x_{t})\geq0$, with equality precisely when the sample losses coincide.

\begin{proposition}[Endogenous price of robustness]
\label{prop:sensitivity}
Let
\begin{equation}
q_{t}^{\star}(\rho_{t}):=\sup_{F_{t}\in\mathcal{P}_{t}(F_{t,N_{t}})}
e^{F_{t}}_{\tau}(L_{x_{t}}),
\qquad
S_{t}(\rho_{t}):=\bigl\{i:\hat{L}_{t,i}>q_{t}^{\star}(\rho_{t})\bigr\},
\qquad
m_{t}:=|S_{t}(\rho_{t})|,
\nonumber
\end{equation}
so that $m_{t}/N_{t}=\hat{P}_{t}(q_{t}^{\star})$. Then $q_{t}^{\star}$ is non-decreasing in $\rho_{t}$, $m_{t}$ is non-increasing in $\rho_{t}$, and on any interval over which $S_{t}$ is constant $q_{t}^{\star}$ is exactly affine in $\rho_{t}$ with
\begin{equation}
\frac{\partial q_{t}^{\star}}{\partial\rho_{t}}
=\frac{(1+\kappa_{\tau})\,\dnorm{x_{t}}}{1+\kappa_{\tau}\,m_{t}/N_{t}}.
\label{eq:sensitivity}
\end{equation}
Consequently $q_{t}^{\star}$ is piecewise affine, non-decreasing and convex in $\rho_{t}$. By contrast, the worst-case CVaR $\alpha\in(0,1)$ on the same ambiguity set satisfies
\begin{equation}
\sup_{F_{t}\in\mathcal{P}_{t}(F_{t,N_{t}})}\mathrm{CVaR}^{F_{t}}_{\alpha}(L_{x_{t}})
=\mathrm{CVaR}^{F_{t,N_{t}}}_{\alpha}(L_{x_{t}})
+\frac{\rho_{t}\dnorm{x_{t}}}{1-\alpha},
\qquad\text{so}\qquad
\frac{\partial}{\partial\rho_{t}}\sup_{F_{t}}\mathrm{CVaR}^{F_{t}}_{\alpha}(L_{x_{t}})
=\frac{\dnorm{x_{t}}}{1-\alpha},
\label{eq:cvar-sensitivity}
\end{equation}
which is constant in $\rho_{t}$ and invariant to the loss distribution.
\end{proposition}

Proposition~\ref{prop:sensitivity} shows that the marginal cost of robustness under the expectile objective depends on the endogenous exceedance frequency $\hat{P}_{t}(q_{t}^{\star})$. A larger exceedance frequency reduces the sensitivity of $q_t^\star$ to the Wasserstein radius $\rho_t$. As $\rho_t$ increases the exceedance frequency may decline and the marginal effect of additional robustness becomes stronger. In contrast the robust CVaR adjustment in \eqref{eq:cvar-sensitivity} is additively separable from the empirical CVaR and its sensitivity to $\rho_t$ is determined only by the confidence level $\alpha$. This difference reflects the broader dependence of expectiles on the distribution of losses compared with quantile-based risk measures \citep{kuan2009assessing,bellini2017risk}.

\begin{proposition}[Degeneracy threshold]
\label{prop:degeneracy}
Suppose $\mathrm{sp}_{t}(x_{t})>0$ and define
\begin{equation}
\rho^{\mathrm{crit}}_{t}(x_{t}):=
\frac{\mathrm{sp}_{t}(x_{t})}{(1+\kappa_{\tau})\,\dnorm{x_{t}}}
=\frac{\max_{i}\hat{L}_{t,i}-\bar{L}_{t}}{(1+\kappa_{\tau})\,\dnorm{x_{t}}}>0 .
\label{eq:rhocrit}
\end{equation}
Then the following hold.
\begin{enumerate}
\item[(i)] $\rho_{t}\geq\rho^{\mathrm{crit}}_{t}(x_{t})$ if and only if
$q_{t}^{\star}\geq\max_{i}\hat{L}_{t,i}$, equivalently
$\hat{P}_{t}(q_{t}^{\star})=0$, in which case the second affine piece of
\eqref{eq:twopiece} is inactive at every sample point.
\item[(ii)] On that region the worst-case expectile is available in closed form,
\begin{equation}
q_{t}^{\star}=(1+\kappa_{\tau})\rho_{t}\dnorm{x_{t}}-\hat{\mu}_{t}^{\top}x_{t},
\qquad
\frac{\partial q_{t}^{\star}}{\partial\rho_{t}}=(1+\kappa_{\tau})\dnorm{x_{t}},
\label{eq:degenerate-q}
\end{equation}
the second identity being the $m_{t}=0$ endpoint of \eqref{eq:sensitivity} and the
maximum attainable value of that derivative.
\item[(iii)] If $\rho_{t}<\rho^{\mathrm{crit}}_{t}(x_{t})$ then
$\hat{P}_{t}(q_{t}^{\star})\geq1/N_{t}>0$ and
$\overline{(\hat{L}_{t}-q_{t}^{\star})^{+}}>0$.
\end{enumerate}
If instead $\mathrm{sp}_{t}(x_{t})=0$, all sample losses coincide,
$\rho^{\mathrm{crit}}_{t}(x_{t})=0$, and \eqref{eq:degenerate-q} holds for every
$\rho_{t}\geq0$.
\end{proposition}

Proposition~\ref{prop:degeneracy} identifies a radius beyond which the model ceases to be an expectile model at all: the tail term vanishes identically and the objective reduces to \eqref{eq:degenerate-q}, in which $\tau$ enters only through the scalar multiplier $1+\kappa_{\tau}$.
\begin{remark}[Relation to the scalar worst-case expectile]
\label{rem:hcm}
Proposition~\ref{prop:degeneracy} admits a direct connection with an existing univariate result. Fix $x_{t}$ and map $\xi\mapsto-\xi^{\top}x_{t}$. The image is a scalar reference distribution $G$ with
$\mu_{G}=-\hat{\mu}_{t}^{\top}x_{t}$ and $\operatorname{ess\,sup}G=\max_{i}\hat{L}_{t,i}$,
and since $\xi\mapsto\xi^{\top}x_{t}$ is Lipschitz with modulus $\dnorm{x_{t}}$ the image of
$\mathcal{P}_{t}(F_{t,N_{t}})$ is contained in the scalar $1$-Wasserstein ball around $G$ of
radius $\varepsilon:=\rho_{t}\dnorm{x_{t}}$. 
Using $1+\kappa_{\tau}=\tau/(1-\tau)$, which corresponds to the parameter $\beta$ in \citet{hu2024extreme}, part~(i) of Proposition~\ref{prop:degeneracy} becomes $\operatorname{ess\,sup}G\leq\mu_{G}+\varepsilon(1+\kappa_{\tau})$ while \eqref{eq:degenerate-q} gives the value
$\mu_{G}+\varepsilon(1+\kappa_{\tau})$.
These expressions coincide with the condition and the corresponding worst-case value in Theorem~2(i) of \citet{hu2024extreme}. 
The contribution of Proposition~\ref{prop:degeneracy} therefore lies not in the scalar identity itself but in its dependence on the portfolio allocation. In particular,
$\rho^{\mathrm{crit}}_{t}$ depends on $x_t$ through both $\mathrm{sp}_{t}(x_{t})$ and $\dnorm{x_{t}}$. This decision dependence provides the basis for the ground-metric comparison in Proposition~\ref{prop:degeneracy}.
\end{remark}

\begin{proposition}[Recovery of the nominal model]
\label{prop:nominal}
Let $\rho_{t}=0$ for every $t$ and let $\lambda\in(0,1)$. Then:
\begin{enumerate}
\item[(i)] $\mathcal{P}_{t}(F_{t,N_{t}})=\{F_{t,N_{t}}\}$, the interval
\eqref{eq:betabounds} collapses to the single point
$\beta_{t,L}=\beta_{t,U}=\hat{\mu}_{t}^{\top}x_{t}$, Case B \eqref{eq:caseB} is empty, and $x_{t}$ is feasible if and only if $\hat{\mu}_{t}^{\top}x_{t}\geq r_{t}$;
\item[(ii)] the right-hand side of \eqref{eq:EXA1-b} is independent of $a_{t,1}$ and equals $\bar{L}_{t}+\kappa_{\tau}\,\overline{(\hat{L}_{t}-q_{t})^{+}}$, whose binding value is the
empirical $\tau$-expectile $e^{F_{t,N_{t}}}_{\tau}(L_{x_{t}})$;
\item[(iii)] the auxiliary variable $w_{t}$ and the $n$ constraints \eqref{eq:EXA1-f} become vacuous;
\item[(iv)] Model \eqref{eq:model-a}--\eqref{eq:model-c} attains the optimal value of the nominal multi-period mean-expectile program
\begin{equation}
\min_{x}\;\sum_{t=1}^{T}\Bigl[-(1-\lambda)\hat{\mu}_{t}^{\top}x_{t}
+\lambda\,e^{F_{t,N_{t}}}_{\tau}(L_{x_{t}})\Bigr]
\quad\text{s.t.}\quad
\hat{\mu}_{t}^{\top}x_{t}\geq r_{t},\quad x_{t}\in\mathcal{X},
\nonumber
\end{equation}
whose per-period risk term admits the scenario-aggregation linear program of
\citet{jakobsons2016scenario}.
\end{enumerate}
\end{proposition}

\begin{proposition}[Norm dependence and implicit concentration penalty]
\label{prop:norm}
The dual norm $\dnorm{x_{t}}$ enters the worst-case value through the term
$(1+\kappa_\tau)\rho_{t}\dnorm{x_{t}}$ and the degeneracy threshold \eqref{eq:rhocrit} through $\dnorm{x_{t}}^{-1}$. Consequently the ground metric in \eqref{eq:wasserstein} determines the implicit regularisation induced by robustification:
\begin{enumerate}
\item[(i)] $\ell_{\infty}$ ground metric (dual $\ell_{1}$): on the simplex
$\norm{x_{t}}_{1}\equiv1$, the robust term is the constant $(1+\kappa_\tau)\rho_{t}$, and $\rho_{\mathrm{crit}}$ is identical for every allocation. Robustification exerts no first-order penalty on concentration, and past the threshold the program reduces to maximization of the sample mean;
\item[(ii)] $\ell_{1}$ ground metric (dual $\ell_{\infty}$): the penalty
$(1+\kappa_\tau)\rho_{t}\norm{x_{t}}_{\infty}$ is a maximum-weight concentration penalty, minimized at equal weights, and remains linear upon introducing $w_{t}\geq x_{t,i}$; $\rho^{\mathrm{crit}}_{t}$ scales with $\norm{x_{t}}_{\infty}^{-1}$;
\item[(iii)] $\ell_{2}$ ground metric: the penalty $(1+\kappa_\tau)\rho_{t}\norm{x_{t}}_{2}$ is of ridge type and the program becomes a second-order cone program.
\end{enumerate}
\end{proposition}

Case~(i) is the configuration adopted by \citet{wang2026multi}. Taken together
with Proposition~\ref{prop:degeneracy}, it implies that, under a long-only
constraint, an $\ell_{\infty}$ ground metric makes the robust term unable to
discriminate between allocations and, beyond $\rho_{\mathrm{crit}}$, leaves
sample-mean maximization as the allocation-relevant component. This motivates
our use of the $\ell_{1}$ ground metric in Section~\ref{sec:ambiguity}.

\section{Empirical analysis}
\label{sec:empirical}
 
\subsection{Data and experimental design}
\label{sec:data}
 
The empirical analysis uses daily simple returns for the $n=90$ constituents of the FTSE index obtained from LSEG Refinitiv for the period from 20 September 2012 to 26 June 2026, comprising $3,591$ trading days. The estimation window is $N_{t}=250$ trading days and the portfolio is rebalanced every $21$ trading days, giving $160$ rebalancing days and an out-of-sample record of $3\,341$ days, or $13.3$ years. Constituents are taken as of the end of the sample and held fixed. This approach may introduce survivorship bias and consequently overstate the reported returns; however, the bias is common to all competing models and therefore does not affect their relative comparison.
 
At each rebalancing date the estimation window is used to form
$\hat\mu_{t}$ and the scenario set $\{\hat\xi_{t,i}\}_{i=1}^{250}$; the period
subproblem is solved; and the resulting weights are held, drifting with
realized returns, until the next rebalancing date. The return target is $r_{t}=\hat\mu_{t}^{\top}(\tfrac1n e)$, the equally weighted mean, so the model must at least match a naive allocation in expectation. Results are reported gross of transaction costs. We set $\tau=\alpha=0.95$ and report
$\lambda\in\{0.10,0.15,0.20\}$ and $\rho_{t}\in\{0.001,0.005,0.010\}$. The
choice of radius is justified in Section~\ref{sec:matching}. Supplementary Section S1 reports the resulting sensitivity analysis.
 
\subsection{Matching risk levels and ambiguity radii}
\label{sec:matching}
 
Two matching problems must be resolved before any cross-model comparison is
meaningful.
\paragraph{Risk level}
Comparing $e_{\tau}$ at $\tau=0.95$ with $\mathrm{CVaR}_{\alpha}$ at
$\alpha=0.95$ is not automatic, since the two functionals are not on a common
scale. We report both functionals for every model so that the reader may judge the comparison directly.
\paragraph{Ambiguity radius}
A common nominal $\rho$ is not a common ambiguity level. Differentiating
each objective with respect to $\rho_{t}$ at fixed $x_{t}$ gives the
coefficient $c(\lambda)$ that multiplies $\rho_{t}\dnorm{x_{t}}$:
\begin{equation}
c_{\text{exp}}=1+\lambda\kappa_{\tau},\qquad
c_{\text{CVaR}}=1+\tfrac{\lambda}{1-\alpha},\qquad
c_{\text{MAD}}=1,\qquad
c_{\text{TSDR}}=1+\tfrac{\lambda}{1-\alpha}+p .
\label{eq:ceff}
\end{equation}
At $\tau=\alpha=0.95$ the expectile and CVaR coefficients differ by $3.5\%$
($2.80$ against $2.90$ at $\lambda=0.1$), so those two families require no
rescaling. Distributionally robust mean absolute deviation has $c=1$ and is
therefore run on the rescaled grid $\rho\cdot c_{\text{exp}}$; in the two-stage
model shortfall penalty fixed at $p=0.1$. Table~\ref{tab:ceff} records the
coefficients and the resulting grids. Throughout we report the effective
radius $c\rho$ alongside the nominal one.
 
\begin{table}[htbp]
\centering
\caption{Effective robustification coefficients from \eqref{eq:ceff} and the
radius grids used. A common $c\rho$ is the matched quantity; a common $\rho$ is
not.}
\label{tab:ceff}
\small
\begin{tabular}{@{}lcccl@{}}
\toprule
Model & $c(0.10)$ & $c(0.15)$ & $c(0.20)$ & $\rho$ grid at $\lambda=0.10$ \\
\midrule
Proposed & 2.80 & 3.70 & 4.60 & $0.001,\ 0.005,\ 0.010$ \\
DR-CVaR-Mean ($\ell_{1},\ell_{\infty}$) & 2.90 & 3.85 & 4.80 & $0.001,\ 0.005,\ 0.010$ \\
DR-MAD & 1.00 & 1.00 & 1.00 & $0.0028,\ 0.0140,\ 0.0280$ \\
Two-stage DR ($p=0.1$) & 3.00 & 3.95 & 4.90 & $0.0009,\ 0.0045,\ 0.0090$ \\
\bottomrule
\end{tabular}
\end{table}
 
\subsection{Benchmark models}
\label{sec:benchmarks}
 
DR-CVaR-Mean under the $\ell_{1}$ ground metric shares the ambiguity set, radius, return target, rebalancing interval and constraint set with the proposed model and differs only in the risk functional: this is reported in Section~\ref{sec:h0}.
DR-CVaR-Mean ($\ell_{\infty}$) is the same model under the $\ell_{\infty}$ ground metric, and the gap between the two models is reported in Section~\ref{sec:h2}. MEPO and MCVaRPO are the nominal ($\rho=0$) counterparts. DR-MAD, Two-stage DR and DR-Moment represent alternative robust models; Mean-Variance, Minimum Variance, $1/N$ and $1/N$ buy-and-hold are the standard non-robust models. Following \citet{demiguel2009optimal}, the equally weighted portfolio serves as a stringent benchmark for evaluating the performance of the proposed model.

\subsection{Main results}
\label{sec:mainresults}
 
Table~\ref{tab:main-ours} reports the proposed model across the nine
$(\rho_{t},\lambda)$ cells together with the theoretical diagnostics of
Sections~\ref{sec:beta} and~\ref{sec:properties}.
Three features are worth noting before any comparison. First, every rebalance
solved and the linearization was exact at every one, so Remark~\ref{rem:w}'s
conservative regime never arose. Second, $\rho_{t}/\rho_{\mathrm{crit}}$ never
exceeds $0.34$, so the model operates well inside the region where the tail
term is active and Proposition~\ref{prop:degeneracy} does not hinder. Third, the
Case A share falls from unity to $0.21$ as $\rho_{t}$ and $\lambda$ increase, 
which means that the intended mean--risk trade-off becomes inoperative more
often. Supplementary Section~S3 examines this diagnostic and its association
with realized performance.
\begin{table}[htbp]
\centering
\caption{Proposed model, FTSE 2012--2026, $3\,341$ out-of-sample days.
$\rho/\rho_{\mathrm{crit}}$ is the degeneracy ratio of
Proposition~\ref{prop:degeneracy}; $s_{A}$ is the share of rebalancing dates in
Case A, where $\lambda$ influences the allocation (Remark~\ref{rem:caseB});
$w$-tight is the share on which the linearization $w_{t}=\norm{x_{t}}_{\infty}$
is exact (Remark~\ref{rem:w}). }
\label{tab:main-ours}
\footnotesize
\begin{tabular}{@{}llrrrrrrrrrr@{}}
\toprule
$\rho$ & $\lambda$ & $\rho/\rho_{\mathrm{crit}}$ & $s_{A}$ & $w$-tight &
Ann.\ ret. & Ann.\ vol. & Sharpe & MaxDD & $\mathrm{CVaR}_{.95}$ & $e_{.95}$ &
Eff.\ $n$ \\
\midrule
0.001 & 0.10 & 0.124 & 1.00 & 1.00 & 0.2062 & 0.2052 & 1.005 & 0.4441 & 0.0293 & 0.0142 &  4.98 \\
0.001 & 0.15 & 0.105 & 1.00 & 1.00 & 0.1613 & 0.1721 & 0.937 & 0.3952 & 0.0252 & 0.0122 &  7.14  \\
0.001 & 0.20 & 0.093 & 1.00 & 1.00 & 0.1441 & 0.1532 & 0.941 & 0.3724 & 0.0226 & 0.0109 &  9.29 \\
0.005 & 0.10 & 0.172 & 1.00 & 1.00 & 0.1394 & 0.1440 & 0.968 & 0.3382 & 0.0216 & 0.0105 & 18.64  \\
0.005 & 0.15 & 0.183 & 0.91 & 1.00 & 0.1166 & 0.1345 & 0.867 & 0.3382 & 0.0202 & 0.0098 & 20.81 \\
0.005 & 0.20 & 0.208 & 0.73 & 1.00 & 0.1007 & 0.1286 & 0.783 & 0.3268 & 0.0195 & 0.0094 & 20.95 \\
0.010 & 0.10 & 0.251 & 0.81 & 1.00 & 0.1127 & 0.1354 & 0.832 & 0.3307 & 0.0204 & 0.0098 & 28.77 \\
0.010 & 0.15 & 0.308 & 0.42 & 1.00 & 0.0812 & 0.1268 & 0.641 & 0.3137 & 0.0191 & 0.0092 & 26.99  \\
0.010 & 0.20 & 0.336 & 0.21 & 1.00 & 0.0674 & 0.1242 & 0.543 & 0.3043 & 0.0185 & 0.0090 & 25.53  \\
\midrule
\multicolumn{2}{@{}l}{$1/N$} & --- & --- & --- & 0.0946 & 0.1541 & 0.614 & 0.3764 & 0.0230 & 0.0110 & 90.00  \\
\bottomrule
\end{tabular}
\end{table}

The cell $(\rho,\lambda)=(0.005,0.10)$ dominates the equally weighted benchmark
on every dimension simultaneously: annualized return $13.9\%$ against $9.5\%$
($+47\%$), volatility $14.4\%$ against $15.4\%$ ($-7\%$), maximum drawdown
$33.8\%$ against $37.6\%$ ($-10\%$), $\mathrm{CVaR}_{0.95}$ $0.0216$ against
$0.0230$ ($-6\%$) and $e_{0.95}$ $0.0105$ against $0.0110$ ($-5\%$), while
holding Eff $n =18.6$. Six of the nine cells dominate $1/N$ on all
four of return, volatility, drawdown and CVaR.

Table~\ref{tab:main-bench} compares the benchmark models. For each robust model family we report the specification with $\lambda=0.10$ and effective radius $c_{\text{exp}}\times0.005$. We also include the parameter-free models. The highest Sharpe ratio in Table~\ref{tab:main-bench} observed in MEPO, the nominal ($\rho=0$) mean-expectile portfolio, at $1.209$ against $0.968$ for the proposed model.
However, MEPO achieves that ratio by holding $2.07$ effective assets with a mean largest weight of $57.6\%$. Its maximum drawdown is $50.6\%$, its
$\mathrm{CVaR}_{0.95}$ is $0.0385$---$67\%$ above the equally weighted
benchmark and $78\%$ above the robust model. The same holds for MCVaRPO. Supplementary Section S4 further compares models for diversification.

Figure~\ref{fig:wealth} complements the terminal statistics by showing how the
differences accumulate through time. All models experience the same major
market drawdown near the middle of the evaluation period, but the proposed
portfolio subsequently recovers more strongly and its wealth difference widens
over the latter part of the sample. The figure provides a descriptive comparison only. Statistical significance is assessed using the dependence-robust inference reported in Supplementary Section~S2.
\begin{table}[htbp]
\centering
\caption{Benchmark comparison at $\lambda=0.10$ and matched effective radius
$c\rho=0.014$ where applicable. Eff.\ $n$ is $1/\mathrm{HHI}$; max $w$ is the
mean largest weight.}
\label{tab:main-bench}
\footnotesize
\begin{tabular}{@{}lrrrrrrrr@{}}
\toprule
Model & Ann.\ ret. & Ann.\ vol. & Sharpe & MaxDD & Calmar &
$\mathrm{CVaR}_{.95}$ & Eff.\ $n$ & max $w$  \\
\midrule
Proposed ($\rho=0.005$) & 0.1394 & 0.1440 & 0.968 & 0.3382 & 0.412 & 0.0216 & 18.64 & 0.055 \\
DR-CVaR-Mean, $\ell_{1}$ ($\rho=0.005$) & 0.0719 & 0.1260 & 0.571 & 0.3183 & 0.226 & 0.0189 & 18.59 & 0.062  \\
DR-CVaR-Mean, $\ell_{\infty}$ ($\rho=0.005$) & 0.0654 & 0.1675 & 0.390 & 0.3478 & 0.188 & 0.0244 &  7.15 & 0.224 \\
DR-MAD ($c\rho=0.014$) & 0.1315 & 0.1457 & 0.903 & 0.3465 & 0.380 & 0.0219 & 28.86 & 0.035  \\
Two-stage DR ($c\rho=0.014$) & 0.1120 & 0.1308 & 0.857 & 0.3183 & 0.352 & 0.0198 & 25.55 & 0.041 \\
DR-Moment & 0.1150 & 0.1299 & 0.885 & 0.3256 & 0.353 & 0.0195 & 10.31 & 0.183  \\
\midrule
MEPO (nominal expectile) & 0.3269 & 0.2705 & 1.209 & 0.5060 & 0.646 & 0.0385 &  2.07 & 0.576  \\
MCVaRPO (nominal CVaR) & 0.2124 & 0.1989 & 1.068 & 0.4358 & 0.487 & 0.0283 &  4.04 & 0.357  \\
Mean--variance & 0.0719 & 0.1166 & 0.617 & 0.3112 & 0.231 & 0.0179 & 11.82 & 0.166  \\
Minimum variance & 0.0675 & 0.1163 & 0.581 & 0.3147 & 0.215 & 0.0178 & 11.80 & 0.166 \\
$1/N$ & 0.0946 & 0.1541 & 0.614 & 0.3764 & 0.251 & 0.0230 & 90.00 & 0.011  \\
$1/N$ buy-and-hold & 0.0952 & 0.1590 & 0.599 & 0.3950 & 0.241 & 0.0235 & 50.98 & 0.068  \\
\bottomrule
\end{tabular}
\end{table}
 
\begin{figure}[h!]
\centering
\includegraphics[width=0.75\textwidth]{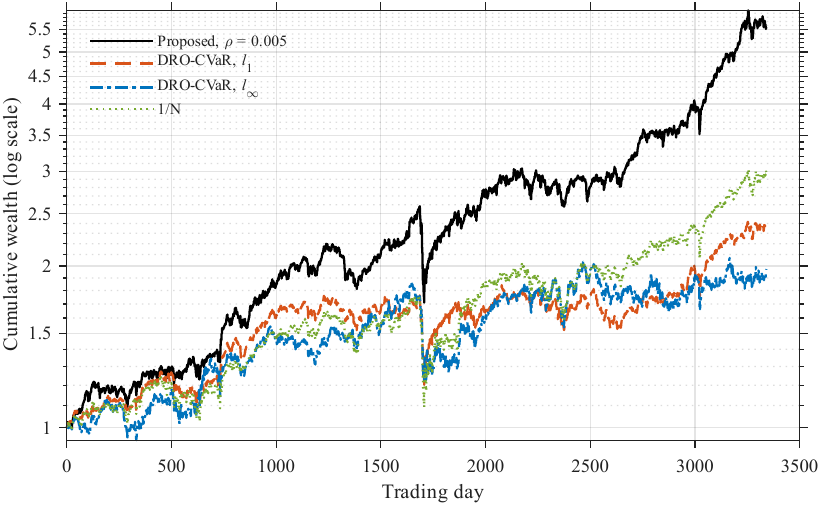}
\caption{Cumulative out-of-sample wealth, FTSE 2012--2026. Proposed model at
$(\rho,\lambda)=(0.005,0.10)$, DR-CVaR under both ground metrics, and the
equally weighted benchmark.}
\label{fig:wealth}
\end{figure}
 
\subsection{Risk-measure comparison: the expectile against Conditional Value-at-Risk}
\label{sec:h0}
 
The comparison in Table~\ref{tab:h0} keeps the ambiguity set, ground metric, radius, trade-off weight, return target, rebalancing interval and constraint set fixed across models. The only difference is the choice of risk functional.

\begin{table}[h!]
\centering
\caption{Proposed model against DR-CVaR-Mean under the $\ell_{1}$ ground
metric, at matched $(\rho,\lambda)$. Positive $\Delta$ favours the expectile.}
\label{tab:h0}
\small
\begin{tabular}{@{}llrrrrrrr@{}}
\toprule
& & \multicolumn{2}{c}{Sharpe} & \multicolumn{2}{c}{Ann.\ return} &
\multicolumn{2}{c}{Eff.\ $n$} & \\
\cmidrule(lr){3-4}\cmidrule(lr){5-6}\cmidrule(lr){7-8}
$\rho$ & $\lambda$ & expectile & CVaR & expectile & CVaR & expectile & CVaR &
$\Delta$Sharpe \\
\midrule
0.001 & 0.10 & 1.005 & 0.927 & 0.2062 & 0.1393 &  4.98 & 10.04 & $+0.078$ \\
0.001 & 0.15 & 0.937 & 0.839 & 0.1613 & 0.1143 &  7.14 & 12.02 & $+0.098$ \\
0.001 & 0.20 & 0.941 & 0.829 & 0.1441 & 0.1056 &  9.29 & 12.25 & $+0.112$ \\
0.005 & 0.10 & 0.968 & 0.571 & 0.1394 & 0.0719 & 18.64 & 18.59 & $+0.397$ \\
0.005 & 0.15 & 0.867 & 0.365 & 0.1166 & 0.0456 & 20.81 & 16.70 & $+0.502$ \\
0.005 & 0.20 & 0.783 & 0.474 & 0.1007 & 0.0603 & 20.95 & 16.32 & $+0.310$ \\
0.010 & 0.10 & 0.832 & 0.403 & 0.1127 & 0.0519 & 28.77 & 23.19 & $+0.429$ \\
0.010 & 0.15 & 0.641 & 0.458 & 0.0812 & 0.0615 & 26.99 & 22.46 & $+0.183$ \\
0.010 & 0.20 & 0.543 & 0.462 & 0.0674 & 0.0621 & 25.53 & 22.44 & $+0.081$ \\
\midrule
\multicolumn{8}{@{}l}{Mean advantage} & $+0.243$ \\
\bottomrule
\end{tabular}
\end{table}
 
The proposed model yields the higher Sharpe ratio in all nine cells. Because the daily returns are dependent, Supplementary Section~S2 also reports two-sided, studentized circular block-bootstrap tests following \citet{ledoit2008robust}. None of the three differences at $\rho=0.001$ is significant at the $5\%$ level, whereas all six differences at $\rho\in\{0.005,0.010\}$ are significant. The results therefore indicate a clearer distinction between the two models as the ambiguity radius increases. They do not, however, imply that the difference is statistically detectable in a neighborhood of the nominal model.

At $(\rho,\lambda)=(0.005,0.10)$ the two models hold almost exactly the same number of effective assets, $18.64$ against $18.59$, yet the expectile portfolio earns $13.9\%$ a year against $7.2\%$. At equal diversification, the proposed model selects a less risky portfolio.
 
\subsection{Ground metric comparison and degeneracy}
\label{sec:h2}
 
Table~\ref{tab:wang} reports the performance of the DR-CVaR-Mean model under the $\ell_{\infty}$ ground metric. For $\rho\geq0.005$, the realized multiplier equals $-\Delta/2$ at all $160$ rebalancing dates, corresponding to the minimum of $\gamma_{2}^{0}$. Moreover, $\beta_{t,L}=r_t$ throughout, so the Case~A share is zero and neither $\hat{\mu}_{t}$ nor $\rho_t$ enters the allocation-relevant part of the objective. Consequently, six parameter combinations yield the same portfolio despite a fivefold increase in $\rho$ and a twofold increase in $\lambda$. In this region, the resulting allocation is therefore insensitive to both parameters.
 
This behavior follows from Proposition~\ref{prop:norm}(i) together with Proposition~\ref{prop:degeneracy}. Under the $\ell_{\infty}$ ground metric we have $\norm{x_t}_{*}= 1$ on the simplex. Hence the robust term becomes constant across feasible portfolios and does not distinguish between alternative allocations. The associated degeneracy threshold $\rho_{\mathrm{crit}}$ is therefore identical for all portfolios. This limits the ability of the robust term to influence the allocation decision.

The empirical results are consistent with this observation. The model based on the $\ell_{\infty}$ ground metric does not outperform the equally weighted benchmark in any of the nine cases. In contrast the proposed model based on the $\ell_{1}$ ground metric outperforms the benchmark in eight of the nine cases and performs better on all four evaluation criteria in six cases. These results indicate that the choice of ground metric has a direct effect on whether the distributionally robust component remains effective in distinguishing among portfolio allocations.

 \begin{table}[htbp]
\centering
\caption{DR-CVaR-Mean under the $\ell_{\infty}$ ground metric.
$a^{\star}$ is the realised dual multiplier and $s_{A}$ the Case A share. The
kink of $\gamma_{2}^{0}$ lies at $-\Delta/2=-10\lambda$.}
\label{tab:wang}
\small
\begin{tabular}{@{}llrrrrrrr@{}}
\toprule
$\rho$ & $\lambda$ & $a^{\star}$ & $-\Delta/2$ & $s_{A}$ & Ann.\ ret. &
Sharpe & $\mathrm{CVaR}_{.95}$ & Eff.\ $n$ \\
\midrule
0.001 & 0.10 & $-21.94$ & $-1.0$ & 0.04 & 0.0769 & 0.478 & 0.0234 & 7.25 \\
0.001 & 0.15 & $-19.98$ & $-1.5$ & 0.00 & 0.0582 & 0.378 & 0.0228 & 8.25 \\
0.001 & 0.20 & $-19.34$ & $-2.0$ & 0.00 & 0.0582 & 0.378 & 0.0228 & 8.25 \\
0.005 & 0.10 & $-1.000$ & $-1.0$ & 0.00 & 0.0654 & 0.390 & 0.0244 & 7.15 \\
0.005 & 0.15 & $-1.500$ & $-1.5$ & 0.00 & 0.0654 & 0.390 & 0.0244 & 7.15 \\
0.005 & 0.20 & $-2.000$ & $-2.0$ & 0.00 & 0.0654 & 0.390 & 0.0244 & 7.15 \\
0.010 & 0.10 & $-1.000$ & $-1.0$ & 0.00 & 0.0654 & 0.390 & 0.0244 & 7.15 \\
0.010 & 0.15 & $-1.500$ & $-1.5$ & 0.00 & 0.0654 & 0.390 & 0.0244 & 7.15 \\
0.010 & 0.20 & $-2.000$ & $-2.0$ & 0.00 & 0.0654 & 0.390 & 0.0244 & 7.15 \\
\bottomrule
\end{tabular}
\end{table}

\section{Conclusions}
\label{sec:conclusions}
We developed a multi-period portfolio model that trades expected return against the worst-case expectile over a Wasserstein ball. The main methodological result is an envelope theorem that converts the implicit worst-case expectile into the unique root of a robust expectation with a two-piece affine integrand. Standard Wasserstein duality then yields four parametric linear programs with $O(N)$ constraints. This reformulation makes an elicitable coherent risk measure computationally available in a framework previously dominated by CVaR.

The analysis also clarifies when robustness changes the allocation. Its marginal price under an expectile is moderated by an endogenous exceedance frequency, whereas the corresponding CVaR penalty is fixed by the confidence level. A decision-dependent critical radius characterizes the point at which the empirical tail component becomes inactive. Beyond this threshold, the ground metric determines the limiting portfolio geometry: concentration under an $\ell_\infty$ ground metric and diversification under an $\ell_1$ ground metric for long-only portfolios. In the experiment, the expectile model has a higher Sharpe ratio than the fully matched CVaR model in
all nine parameter cells. The differences are statistically significant in all six cells with $\rho\in\{0.005,0.010\}$, but not in the three cells closest to the nominal distribution. In six of the nine parameter settings, the expectile model simultaneously outperforms the equally weighted portfolio in terms of return, volatility, drawdown, and CVaR. The $\ell_\infty$ norm exhibits the predicted degeneracy, as six distinct parameter combinations result in an identical portfolio allocation. For implementation, the radius should be reported relative to
$\rho_{\mathrm{crit}}$, not only in absolute units, and the share of Case A rebalancing dates should accompany performance statistics because it records how often the intended mean--risk trade-off is active. The ground metric should likewise be treated as an economic modeling choice. In particular, an
$\ell_\infty$ metric makes the first-order robust term constant on the long-only simplex; an $\ell_1$ metric instead penalizes the largest allocation.

The evidence is limited to a single equity market and is based on end-of-sample constituents, so survivorship bias and market-specific effects cannot be ruled out. The base comparison uses gross returns, although the supplement shows that the key results remain robust to the transaction-cost levels considered.
The theory also assumes long-only allocations, an $\ell_1$ Wasserstein ball, and a separable period-wise objective. Future work could integrate turnover and short-selling constraints, test other markets, and study nested conditional expectiles.

\appendix
\section{Proofs}\label{appen}
\subsection{Proofs of results in Section~\ref{sec:model}}

\subsubsection{Lemma~\ref{lem:timeconsistency}}
\begin{proof}
Monotonicity of $e_{\tau}$ holds for every $\tau\in(0,1)$ by property~(b), and monotonicity of the expectation operator is immediate; a convex combination of monotone mappings is monotone. Coherence on $\tau\in[\tfrac12,1)$ follows from properties (a)--(c), the expectation being coherent and the class of coherent risk measures being closed under convex combination. Monotonicity is the sole hypothesis required by Theorem~3 of \citet{chen2013optimal}, whence the
separable decomposition \eqref{eq:separable} applies and time consistency follows.
\end{proof}

\subsubsection{Lemma~\ref{lem:nonrep}}

\begin{proof}
Comonotone additivity of a law-invariant coherent risk measure is equivalent, by
\citet{kusuoka2001law}, to representability as a mixture of CVaR functionals. \citet{bellini2014generalized} show that among generalized quantiles the expectile fails comonotone additivity for every $\tau\neq\tfrac12$; the claim follows.
\end{proof}

\subsection{Proofs of results in Section~\ref{sec:reformulation}}
\subsubsection{Lemma~\ref{lem:fixedpoint}}
\begin{proof}
Write $e=e^{F}_{\tau}(L)$. Definition \eqref{eq:expectile-def} states
$\tau\mathbb{E}_{F}(L-e)^{+}=(1-\tau)\mathbb{E}_{F}(e-L)^{+}$. Since
$\mathbb{E}_{F}(L-e)^{+}-\mathbb{E}_{F}(e-L)^{+}=\mathbb{E}_{F}[L]-e$, substitution gives
$\tau\mathbb{E}_{F}(L-e)^{+}=(1-\tau)\bigl(\mathbb{E}_{F}(L-e)^{+}-\mathbb{E}_{F}[L]+e\bigr)$, hence
$(2\tau-1)\mathbb{E}_{F}(L-e)^{+}=(1-\tau)\bigl(e-\mathbb{E}_{F}[L]\bigr)$. Dividing by $1-\tau>0$ and
rearranging yields \eqref{eq:fixedpoint}.
\end{proof}

\subsubsection{Proof of Theorem~\ref{thm:envelope}}
\label{app:envelope}
\begin{proof}
Fix $x$, write $L=L_x$ and
$g(q;F):=\mathbb{E}_F[L]+\kappa_\tau\,\mathbb{E}_F(L-q)^+-q$, so that
$\Psi=\sup_{F\in\mathcal{P}}g(\cdot\,;F)$.

\emph{(i) A slope bound uniform in $F$.} For each $F\in\mathcal{P}$ the map
$q\mapsto\mathbb{E}_F(L-q)^+$ is non-increasing and $1$-Lipschitz, so for $q_1<q_2$
\[
-(1+\kappa_\tau)(q_2-q_1)\ \le\ g(q_2;F)-g(q_1;F)\ \le\ -(q_2-q_1),
\]
the two bounds being independent of $F$. Each $g(\cdot\,;F)$ is therefore strictly decreasing, with unique root $e^F_\tau(L)$ by Lemma~\ref{lem:fixedpoint}; and since the bounds are uniform, the
pointwise supremum $\Psi$ inherits both the $(1+\kappa_\tau)$-Lipschitz continuity and the strict decrease at rate at least one. Finiteness of $\Psi$ follows from $|L|\le\|\xi\|\,\|x\|_*$ and $(L-q)^+\le|L|+|q|$ together with $\sup_{F\in\mathcal{P}}\mathbb{E}_F\|\xi\|<\infty$. As $\Psi(q)\to+\infty$ for $q\to-\infty$ and $\Psi(q)\to-\infty$ for $q\to+\infty$, continuity and
strict monotonicity give a unique root $q^\star$.

\emph{(ii) Identification of the root.} For every $F\in\mathcal{P}$ we have
$g(q^\star;F)\le\Psi(q^\star)=0$, and since $g(\cdot\,;F)$ is strictly decreasing with root $e^F_\tau(L)$ this forces $e^F_\tau(L)\le q^\star$; hence
$\sup_{F}e^F_\tau(L)\le q^\star$. Conversely, fix $q<q^\star$. Strict monotonicity gives $\Psi(q)>0$, so by definition of the supremum there is $F\in\mathcal{P}$ with $g(q;F)>0$ and therefore $e^F_\tau(L)>q$. Letting $q\uparrow q^\star$ yields the reverse inequality, so $\sup_{F}e^F_\tau(L)=q^\star$, and $q^\star=\min\{q:\Psi(q)\le0\}$ because $\Psi$ is continuous and strictly decreasing. The supremum is attained whenever $\mathcal{P}$ is weakly compact and
$F\mapsto\mathbb{E}_F[L]+\kappa_\tau\mathbb{E}_F(L-q^\star)^+$ is upper semicontinuous, which holds for the \eqref{eq:ambiguityset} under \eqref{eq:moment}.
\end{proof}
\subsubsection{Lemma~\ref{lem:twopiece}}
\begin{proof}
Immediate from $(\cdot)^{+}=\max\{0,\cdot\}$: the first branch corresponds to
$-\xi_{t}^{\top}x_{t}-q_{t}\leq0$ and the second to $-\xi_{t}^{\top}x_{t}-q_{t}>0$.
\end{proof}

\subsubsection{Lemma~\ref{lem:betaL}}
\begin{proof}
$D_{t}(q,\beta)$ is the optimal value of a linear program in $F_{t}$ whose right-hand side is perturbed by $\beta$, hence concave in $\beta$. By Lemma~\ref{lem:twopiece} both affine pieces of $h_{t}$ carry a negative coefficient on $\xi_{t}^{\top}x_{t}$, so $h_{t}$ is non-increasing in $\xi_{t}^{\top}x_{t}$, and, since $h_{t}$ depends on $\xi_{t}$ only through $\xi_{t}^{\top}x_{t}$, every maximizing sequence for $\mathbb{E}_{F_{t}}[h_{t}]$ over $\mathcal{P}_{t}$ spends the whole budget $\rho_{t}$ reducing
$\xi_{t}^{\top}x_{t}$ at the maximal rate $w_{t}$ per unit of $\norm{\cdot}_{1}$
transport, so its mean tends to
$\hat{\mu}_{t}^{\top}x_{t}-\rho_{t}w_{t}\leq\beta_{t,L}$. Concavity then gives $D_{t}$ non-increasing on $[\beta_{t,L},\beta_{t,U}]$, whence $\gamma_{t,1}\leq0$ and $a_{t,1}=1+\gamma_{t,1}\leq1$. Substituting $a_{t,1}-1\leq0$ and $\kappa_{\tau} P\geq0$ into \eqref{eq:slope} yields \eqref{eq:slope-negative}.
\end{proof}

\subsubsection{Proof of Theorem~\ref{thm:reformulation}}
\begin{proof}
\emph{(i) The inner dual.} By Lemma~\ref{lem:twopiece} the integrand is
$h_t(\xi)=\max\{-k_1\xi^{\top}x_t+c_1q_t,\,-k_2\xi^{\top}x_t+c_2q_t\}$ with
$(k_1,k_2,c_1,c_2)=(1,1+\kappa_\tau,0,-\kappa_\tau)$, a maximum of two affine functions and hence convex and Lipschitz. Wasserstein duality \citep{mohajerin2018data, gao2023distributionally,zhang2025short} applied to
\eqref{eq:inner-a}--\eqref{eq:inner-c}, with multiplier $\gamma_{t,1}$ on the mean constraint and $\gamma_{t,2}\ge0$ on the transport constraint, yields \eqref{eq:dual-a}--\eqref{eq:dual-d}. Strong duality requires $\beta_t$ interior to the
attainable mean interval, which holds in Case~B by \eqref{eq:caseB}. In Case~A the selected endpoint is extreme and Slater fails, but $D_t(q,\cdot)$ is concave by Lemma~\ref{lem:betaL} and the optimal transport value is continuous in the mean constraint, so the dual value is continuous at the endpoint and the substitution remains exact.

\emph{(ii) Factorization of $\gamma_{t,2}^\star$.} Under the $\ell_1$ ground metric the dual norm is $\ell_\infty$, so \eqref{eq:dual-d} read $|a_{t,1}|w_t\le\gamma_{t,2}$ and
$|a_{t,1}+\kappa_\tau|w_t\le\gamma_{t,2}$ with $w_t=\|x_t\|_\infty>0$. Since $\rho_t>0$ the multiplier binds at its smallest admissible value, giving $\gamma_{t,2}^\star=\gamma_{t,2}^0(a_{t,1})w_t$ with
$\gamma_{t,2}^0$ as in \eqref{eq:gamma2}. The factor $w_t$ leaves the kink at $a_{t,1}=-\kappa_\tau/2$ unchanged, and $\gamma_{t,2}^0\ge\kappa_\tau/2>0$ throughout, so no separate dual-feasibility restriction on $a_{t,1}$ is required.

\emph{(iii) The fixed point as an inequality.} By Theorem~\ref{thm:envelope} the worst-case expectile is $\min\{q:D_t(q,\beta_t)\le q\}$. Encoding this as \eqref{eq:EXA1-b} or \eqref{eq:EXB-b} is exact because $q_t$ enters the objective with strictly positive coefficient $\lambda$, so the minimization drives it to the boundary of the constraint, which is the root.

\emph{(iv) Exhaustiveness and exactness.} By Lemma~\ref{lem:betaL} only $\beta_{t,L}$ arises, and by \eqref{eq:betabounds} it takes one of two forms according to whether $\hat{\mu}_t^{\top}x_t-\rho_tw_t\ge r_t$. Crossing this with the two regimes of $\gamma_{t,2}^0$ gives four mutually exclusive and collectively exhaustive branches, within each of which the substitutions of $\gamma_{t,2}^\star$ and $\beta_t^\star$ are exact by (ii) and Lemma~\ref{lem:betaL}; the optimal value
is the minimum over the four. Exactness of the linearization $w_t^\star=\|x_t\|_\infty$ on each case region, and its conservatism outside, is Remark~\ref{rem:w}. 
\end{proof}

\subsection{Proofs of propositions in Section~\ref{sec:properties}}
\subsubsection{Proof of Proposition~\ref{prop:sensitivity}}
\begin{proof}
\emph{(i) The robust fixed point.} By Lemma~\ref{lem:twopiece} the integrand $h_t$ in \eqref{eq:Psi} is a maximum of two affine functions of $\xi_t$ with gradients $-x_t$ and $-(1+\kappa_\tau)x_t$, hence Lipschitz with modulus $(1+\kappa_\tau)\|x_t\|_*$. Without the mean constraint the dual \eqref{eq:dual-d} at $\gamma_{t,1}=0$ reads $\gamma_{t,2}\ge(1+\kappa_\tau)\|x_t\|_*$, in agreement with $\gamma_{t,2}^0(1)=1+\kappa_\tau$, and the dual objective is increasing in $\gamma_{t,2}$; the minimum is attained at $\gamma_{t,2}^\star=(1+\kappa_\tau)\|x_t\|_*$, so that by Theorem~\ref{thm:envelope} the quantity $q_t^\star$ is the unique fixed point of
\begin{equation}\label{eq:G}
G(q):=\sup_{F_t\in\mathcal{P}_t}\mathbb{E}_{F_t}[h_t]+q
=(1+\kappa_\tau)\rho_t\|x_t\|_*+\bar{L}_t+\kappa_\tau\overline{(\hat{L}_t-q)^+}.
\end{equation}
Since $G$ is non-decreasing in $\rho_t$ and non-increasing in $q$, while $\Psi=G-q$ is strictly decreasing, $q_t^\star$ is non-decreasing in $\rho_t$; hence $S_t$ shrinks and $m_t$ is non-increasing. Holding $S_t$ fixed gives
$\overline{(\hat{L}_t-q)^+}=N_t^{-1}\sum_{i\in S_t}(\hat{L}_{t,i}-q)$, and $q_t^\star=G(q_t^\star)$ rearranges to
\begin{equation}
q_t^\star\Bigl(1+\tfrac{\kappa_\tau m_t}{N_t}\Bigr)
=(1+\kappa_\tau)\rho_t\|x_t\|_*+\bar{L}_t+\frac{\kappa_\tau}{N_t}\sum_{i\in S_t}\hat{L}_{t,i},\nonumber
\end{equation}
which is affine in $\rho_t$; differentiating yields \eqref{eq:slope}. As $m_t$ is non-increasing in $\rho_t$ the slope is non-decreasing, so $q_t^\star$ is convex.

\emph{(ii) The CVaR comparison.} By the Rockafellar--Uryasev representation \citep{rockafellar2000optimization}, $\mathrm{CVaR}^F_\alpha(L_{x_t})=\min_q\{q+(1-\alpha)^{-1}\mathbb{E}_F(L_{x_t}-q)^+\}$.
The objective is convex in $q$ and linear in $F$ and $\mathcal{P}_t$ is convex, so the supremum over $F_t$ and the minimum over $q$ may be interchanged. For each fixed $q$ the map
$\xi_t\mapsto(-\xi_t^{\top}x_t-q)^+$ is Lipschitz with modulus $\|x_t\|_*$, whence
$\sup_{F_t}\mathbb{E}_{F_t}(L_{x_t}-q)^+=\rho_t\|x_t\|_*+\overline{(\hat{L}_t-q)^+}$. Substituting and minimizing over $q$ gives \eqref{eq:cvar-sensitivity}: the robust term is additively separable from the data and the minimizing $q$ is the empirical Value-at-Risk, independently of $\rho_t$. 
\end{proof}

\subsubsection{Proof of Proposition~\ref{prop:degeneracy}}

\begin{proof}
Put $q_0:=\max_i\hat{L}_{t,i}$ and let $G$ be as in \eqref{eq:G}. Since $\Psi=G-q$ is
strictly decreasing with unique root $q_t^\star$, we have $q_t^\star\ge q_0$ iff $G(q_0)\ge q_0$;
as $\overline{(\hat{L}_t-q_0)^+}=0$ this reads
$(1+\kappa_\tau)\rho_t\|x_t\|_*\ge\mathrm{sp}_t(x_t)$, that is
$\rho_t\ge\rho^{\mathrm{crit}}_t(x_t)$, proving (i). On that region every $\hat{L}_{t,i}\le q_t^\star$, so by Lemma~\ref{lem:twopiece} the second piece is inactive at every sample point, the fixed point reduces to \eqref{eq:degenerate-q}, and its derivative is the $m_t=0$ endpoint of \eqref{eq:slope},
maximal since $m_t\ge0$; this is (ii). If $\rho_t<\rho^{\mathrm{crit}}_t$ then $q_t^\star<q_0$ by (i), so the index attaining $q_0$ has $\hat{L}_{t,i}>q_t^\star$, giving (iii). If
$\mathrm{sp}_t(x_t)=0$ then $\rho^{\mathrm{crit}}_t(x_t)=0$ and the equivalence in (i) holds for every $\rho_t\ge0$.
\end{proof}

\subsubsection{Proof of Proposition~\ref{prop:nominal}}

\begin{proof}
At $\rho_t=0$ the ball is the singleton $\{F_{t,N_t}\}$ and \eqref{eq:betabounds} gives $\beta_{t,L}=\max\{r_t,\hat{\mu}_t^{\top}x_t\}$, $\beta_{t,U}=\hat{\mu}_t^{\top}x_t$; the interval is
non-empty iff $\hat{\mu}_t^{\top}x_t\ge r_t$ and then reduces to that point, so Case~B is empty. This is (i), and (iii) follows since $w_t$ enters only through $\rho_tw_t$.

For (ii), each $s_{t,i}$ carries coefficient $1/N_t>0$ in \eqref{eq:EXA1-b} and so attains the larger of its bounds \eqref{eq:EXA1-a}--\eqref{eq:EXA1-d},
$s^\star_{t,i}=a_{t,1}\hat{L}_{t,i}+\kappa_\tau(\hat{L}_{t,i}-q_t)^+$. Substituting at $\rho_t=0$ and using $\bar{L}_t=-\hat{\mu}_t^{\top}x_t$, the terms in $a_{t,1}$ cancel and the constraint becomes
$q_t\ge\bar{L}_t+\kappa_\tau\overline{(\hat{L}_t-q_t)^+}$---the mean constraint being redundant at $\rho_t=0$, its multiplier is unidentified. Since $\lambda>0$ this binds (Remark~\ref{rem:no-bisection}) and its unique root is $e^{F_{t,N_t}}_\tau(L_{x_t})$ by Lemma~\ref{lem:fixedpoint}. Substituting (i)--(iii) into \eqref{eq:EXA1-a} gives (iv).
\end{proof}

\subsubsection{Proof of Proposition~\ref{prop:norm}}

\begin{proof}
The norm enters Theorem~\ref{thm:reformulation} through three channels only: the term $\gamma_{t,2}^0(a_{t,1})\rho_t\|x_t\|_*$ in the fixed-point constraints, the Kantorovich bound $\rho_t\|x_t\|_*$ underlying \eqref{eq:betabounds}, and the threshold \eqref{eq:rhocrit}, which scales as $\|x_t\|_*^{-1}$. On $\mathcal{X}$ we have $\|x_t\|_1\equiv1$, giving (i); $\|x_t\|_\infty\ge1/n$ and $\|x_t\|_2\ge n^{-1/2}$, each with equality exactly at $x_t=x^{\mathrm{ew}}$ and each attaining $1$ at the vertices, giving the concentration penalties in (ii) and (iii). Linearity in (ii) is Remark~\ref{rem:w}; in (iii) the constraint $\gamma_{t,2}\ge\gamma_{t,2}^0(a_{t,1})\|x_t\|_2$ is a second-order cone constraint once $a_{t,1}$ is fixed, all others being linear.
\end{proof}

\section*{Data availability}
The data sources have been shared in the article.

\section*{Acknowledgements}
The first author is thankful to the financial assistance provided by the Council of Scientific and Industrial Research (CSIR), India (award number - 09/0086(17076)/2023-EMR-I) for carrying this research work. The DST-FIST grant SR/FST/MS-1/2019/45 is acknowledged for the computing facility in the department.

\bibliographystyle{apalike}
\bibliography{sample}
\end{document}